\documentclass[onecolumn,journal]{IEEEtran}
\IEEEoverridecommandlockouts
\usepackage{graphicx}
\usepackage{amsmath,amssymb}
\usepackage{cases}
\usepackage{color}
\usepackage{amsfonts}
\usepackage{indentfirst}
\usepackage{float}
\usepackage{setspace}
\usepackage{cite}
\usepackage{caption}
\usepackage{algorithm}
\usepackage{algpseudocode}
\usepackage{booktabs}
\usepackage{subfigure}
\usepackage{multirow}
\usepackage{amsthm}
\usepackage{diagbox}
\usepackage{textcomp}
\usepackage{stfloats}
\usepackage{verbatim}
\usepackage{graphicx}

\makeatletter
\newtheoremstyle{mythm}{3pt}{3pt}{}{16pt}{\bfseries}{:}{.5em}{}

\theoremstyle{mythm}
\newtheorem{theorem}{Theorem}
\newtheorem{example}{Example}
\newtheorem{definition}{Definition}
\newtheorem{remark}{Remark}

\newtheorem{proposition}{Proposition}
\newtheorem{corollary}{Corollary}
\newtheorem{lemma}{Lemma}
\newtheorem{construction}{Construction}

\begin{document}
	\title{Fundamental Limits of Coded Caching with Fixed Subpacketization}
\author{Minquan Cheng, Yifei Huang,  Youlong Wu,  and Jinyan Wang
	\thanks{This work was presented in part at the 2022 IEEE International Symposium on Modeling and Optimization in Mobile, Ad hoc, and Wireless Networks (WiOpt), in Torino, Italy, Sept 2022\cite{WiOpt'Cheng}. }
	\thanks{M. Cheng, Y. Huang and J. Wang are with the Key Lab of Education Blockchain and Intelligent Technology, Ministry of Education, and also with the Guangxi Key Lab of Multi-source Information Mining $\&$ Security, Guangxi Normal University, 541004 Guilin, China (e-mail: chengqinshi@hotmail.com, huangyifei59@163.com, wangjy612@gxnu.edu.cn). Y. Huang is also with the School of Science, Guilin University of Aerospace Technology, 541004 Guilin, China.
	}
	\thanks{Y. Wu is with the School of Information Science and Technology, ShanghaiTech University, 201210 Shanghai, China (e-mail: wuyl1@shanghaitech.edu.cn).}
}
	
	\maketitle
	
	\begin{abstract}
	Coded caching is a promising technique to create coded multicast opportunities for cache-aided networks. By splitting each file into $F$ equal packets (i.e., the subpacketization level $F$) and letting each user cache a set of packets, the transmission load can be significantly reduced via coded multicasting.  It has been shown that a higher subpacketization level could potentially lead to a lower transmission load, as more packets can be combined for efficient transmission. On the other hand, a larger $F$ indicates a higher coding complexity and is problematic from a practical perspective when $F$ is extremely large. Despite many works attempting to design coded caching schemes with low subpacketization levels, a fundamental problem remains open: What is the minimum transmission load given any fixed subpacketization level?  In this paper, we consider the classical cache-aided networks with identically uncoded placement and one-shot delivery strategy, and investigate the fundamental trade-off between the transmission load and the subpacketization level. We propose a \emph{general} lower bound on the transmission load for any fixed subpacketization by reformulating the centralized coded caching schemes via the combinatorial structure of the corresponding placement delivery array.  The lower bound also recovers existing optimality results for the bipartite graph scheme (including the well-known Maddah-Ali and Niesen (MN) scheme and the conjugate MN scheme) as well as the grouping bipartite graph scheme. Furthermore, by carefully exploiting the combinatorial structure and computing the union size of sorted sets, we establish a new optimality result, i.e.,  the partition scheme can achieve the optimal rate-subpacketization trade-off.
\end{abstract}


\section{Introduction}
\label{sec-introduction}

Caching has been recognized as an effective method to smooth out network traffic during peak times. In cache-aided networks, some content is proactively stored in the users' local cache memories during off-peak hours in the hope that the pre-stored content will be required during peak hours. When this happens, content is retrieved locally, thereby reducing the transmission load from the server to the users. Coded caching was originally proposed by Maddah-Ali and Niesen (MN) in \cite{MN} to further reduce the amount of transmission by creating broadcast coding opportunities, where a central server transmits some coded symbols and each user uses their cache to cancel the non-desired file, and has been widely used in heterogeneous wireless networks.

\subsection{Centralized coded caching system}\label{Model}
Consider a $(K,M,N)$ centralized caching system, see Fig. \ref{fig-origin-system}, where a single server contains a library of $N$ independent files $W_1,\ldots,W_{N}$ each of $T$ bits, and connects to $K$ users via a noiseless shared link. Each user is equipped with a cache memory of size $M$ files, where $N \geq \max\{K,M\}$. Denote the   $N$ files by $\mathcal{W}=\{W_1,\ldots,W_{N}\}$  and the $K$ users by $\mathcal{K}=\{1,\ldots,K\}$.  An $F$-division $(K,M,N)$ coded caching scheme consists of two phases as follows:
\begin{figure}[t] \centering
	\includegraphics[width=4in]{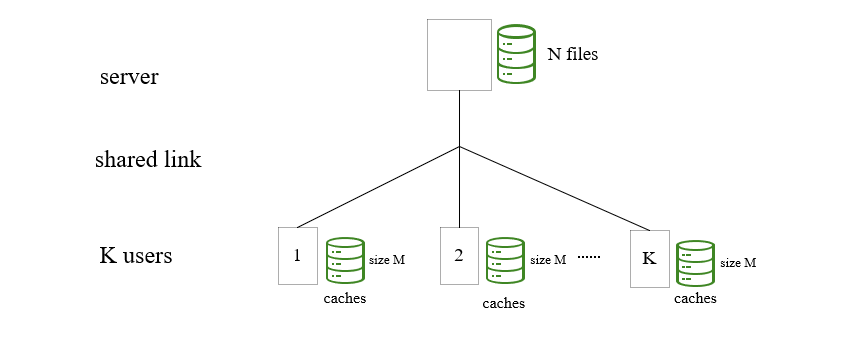} 
	\caption{A $(K,M,N)$ centralized caching  system }\label{fig-origin-system}
\end{figure}
\begin{itemize}
	\item {\bf Placement phase:} In this phase, the server does not know the users' later requests. During the off-peak traffic time, each file is divided into $F$ equal packets\footnote{Memory sharing technique may lead to non-equally divided packets \cite{MN}, and we omit this case for simplicity.}, i.e., $W_i=\{W_{i,j}|j=1,2,\ldots,F\}$, where $W_{i,j}$ represents the $j$-th packet of the file $W_i$. Then each user caches a set of packets or linear combinations of packets from the server. If packets are cached directly, it is called uncoded placement; if linear combinations of packets are cached, we call it coded placement. Let $\mathcal{Z}_k$ denote the contents cached by user $k\in[K]$.  {\color{black}A placement strategy is called {\it identically uncoded} placement when a user caches the $j$-th packet of a file, it will cache the $j$-th packet of all files.}
	
	\item {\bf Delivery phase:} During the peak traffic times, each user randomly requests one file from the files set $\mathcal{W}$ independently. The request vector is denoted by $\mathbf{d}=(d_1,\ldots,d_{K})$, i.e., user $k$ requests the $d_k$-th file $W_{d_k}$, where $d_k\in\{1,\ldots,N\}$ and $k\in\mathcal{K}$. The server broadcasts coded signal $X_t$ as a linear combination of some required packets, where   $t=1,\ldots, S_{{\bf d}}$,   to users such that each user can recover its requested file with the help of its cache contents.  A delivery strategy is called \emph{one-shot} if each user can recover any desired subpacket based on the cached information and at most one transmitted signal, i.e., for each desired subpacket $W_{d_k,j}$, there exists a transmit signal $X_t$, for some $t\in[S]$, satisfying 
	\begin{IEEEeqnarray}{rCl}
		H(W_{d_k,j}|X_t,\mathcal{Z}_k)=0.
	\end{IEEEeqnarray}
\end{itemize}

In this paper, we focus on the $(K,M,N)$ centralized coded caching system with {\it identically uncoded} placement and {\it one-shot} delivery strategy. We define the corresponding transmission load $R$ as the maximal normalized transmission amount among all the requests in the delivery phase, i.e.,
\begin{align*}
	R=\max\{S_{{\bf d}}/F\ | \mathbf{d}\in \{1,\ldots,N\}^K\}.	 
\end{align*} 

We say a pair  $(R,F)$ is achievable if there exists a coded caching scheme with transmission load $R$ and subpacketization $F$ to ensure that all users' request files are successfully recovered.  Previous studies have shown that the transmission load $R$  can be reduced via coded multicasting when the subpacketization level $F$ is sufficiently large. However,  in practice, the file size can be any finite positive number, implying that the subpacketization is finite. Besides, larger subpacketization levels also increase coding complexity and can be impractical for extremely large values of $F$.  We  define the fundamental trade-off between the transmission load and subpacketization level (i.e., rate-subpacketization trade-off) as
\begin{IEEEeqnarray}{rCl}
	R^*(F) = \text{inf} \{ R: (R,F)~\text{is achievable}\}.
\end{IEEEeqnarray}

\subsection{Prior Work}

The first well-known scheme called MN scheme was proposed in  \cite{MN}, whose transmission load is optimal under the constraint of uncoded placement and $N\geq K$~\cite{WTP,YMA2018}, and generally
order optimal within a factor of $2$ \cite{YMA2019}.  
However, there exists a main practical issue of the MN scheme, i.e., its subpacketization increases exponentially with the number of users. To study the finite subpacketization, a matrix with $F$ rows and $K$ columns called placement delivery array (PDA), which can be used to realize a coded caching under identically uncoded placement and one-shot delivery, was proposed in \cite{YCTC}. There are other viewpoints of characterizing coded caching schemes from the viewpoint of graph theory, combinatorial design theory, and coding theory \cite{SZG,YTCC,STD,TR} under identically uncoded caching placement. The authors in \cite{CLTW} introduced a linear characterization of coded caching schemes with coded caching placement from the viewpoint of linear algebra. Clearly, all the schemes under identically uncoded placement and one-shot delivery can be regarded as special cases of linear characterization. From linear characterization, the problem of designing schemes is transformed to constructing three classes of matrices satisfying certain rank conditions that represent user caching strategy, server broadcasting strategy, and user decoding strategy. Furthermore, by these three classes of matrices, the authors showed that the minimum storage regenerating codes with optimal repairing bandwidth proposed by Dimakis et al in \cite{DGWWR} can be used to generate the related linear coded caching scheme. That is the reason why there are many similar structures used between the constructions of coded caching schemes and the minimum storage regenerating codes. Since the study of linear coded caching schemes is much more complicated than that of schemes under uncoded placement, there is only one class of linear coded caching scheme proposed in \cite{CLTW}, which has better performance than the schemes under identically uncoded caching placement for the same subpacketization and memory ratio, while there are many constructions of the schemes under identically uncoded caching placement.

In fact, all coded caching schemes in \cite{SZG,STD,TR,YTCC} can be characterized by PDAs. Very recently, the authors in \cite{CWZW} proposed a unified construction framework for PDA that can represent all the constructions in \cite{YCTC,SZG,TR,YTCC}, and so on. As a result, the problem of designing a PDA can be transformed into a problem of choosing a row index matrix and a column index set appropriately.  There are many other studies on the PDA, such as $3$-vector set \cite{CJTY}, hamming distance \cite{ZCW}, cartesian product \cite{WCWC}, rainbow \cite{XXGL}, binary matrices\cite{SCS}, injective arc coloring of regular digraph \cite{WCCL} and the projective geometry \cite{CHKSM}. In addition, PDA has also been  widely used in various settings such as D2D networks \cite{WCYT}, hierarchical networks \cite{KWC},
combination network \cite{YWY,CLZW}, distributed computing \cite{YYW}, multi-access \cite{SSRB}, hotplug \cite{RCRB}, multi-antenna \cite{YWCQC,NKPR} and location-aware multi-antenna \cite{MHMT}, and so on.


Despite efforts to develop coded caching schemes with low subpacketization levels, the fundamental problem of determining the minimum transmission load for a specific subpacketization level remains an open challenge. Existing converse results in \cite{WTP,YMA2018,YMA2019} determine the minimum communication load without considering subpacketization levels, thus precluding the establishment of a fundamental trade-off between transmission load and subpacketization. Consequently, the lack of an information-theoretic lower bound for transmission load at arbitrary subpacketization levels hinders the evaluation of current low-subpacketization coded caching schemes. To address this issue, this paper investigates the theoretical relationship between transmission load and subpacketization level.

\subsection{Contribution and Organization}
In this paper, we consider an $F$-division $(K,M,N)$ centralized caching system with identically uncoded data placement and one-shot delivery.   We propose a lower bound on the optimal transmission load $R^*(F)$ for any fixed number of users $K$,  memory ratio $M/N$, and subpacketization level $F$. 
Similar to the cut-set bound technique, where individual cuts yield converse bounds and identifying the optimal cut is non-trivial, our lower bound is derived from a sequence of ordered subsets. Determining the optimal ordered subsets is therefore essential for establishing optimality.
By exploiting the combinatorial structure of each placement strategy and carefully computing the union size of the ordered subsets, we show that the partition scheme in \cite{YCTC} can achieve its optimal rate-subpacketization trade-off. Using the same method, we can also show that the bipartite graph scheme in \cite{YTCC,SZG} (including the MN scheme \cite{MN} and the conjugate MN scheme \cite {CJYT}) can achieve their optimal rate-subpacketization trade-offs. This indicates the optimality of the proposed lower bound in many cases.





The rest of this paper is organized as follows. In Section \ref{sec_prob}, the concept of PDA and the corresponding scheme are introduced. In Section \ref{lower-set}, a lower bound on the transmission load for any user number, memory ratio, {\color{black}and subpacketization is derived and the optimality of the existing schemes is shown. Section \ref{sec-optimal-2} presents the proof of our main result.} Finally, Section \ref{conclusion} concludes the paper.

\emph{Notations:}  In this paper,  we use bold capital letters, bold lowercase letters, and calligraphic fonts to denote arrays, vectors, and sets, respectively. If $a$ is not divisible by $q$, $\langle a\rangle_{q}$ denotes the least non-negative residue of $a$ modulo $q$; otherwise, $\langle a\rangle_{q}:=q$.
$|\cdot|$ is used to represent the cardinality of a set or
the length of a vector.
For any positive integers $a$, $b$, $t$ with $a<b$ and $t\leq b $, and any nonnegative set $\mathcal{V}$,
let $[a,b]=\{a,a+1,\ldots,b\}$, especially $[1,b]$ be shorten by $[b]$, and
${[b]\choose t}=\{\mathcal{V}\ |\   \mathcal{V}\subseteq [b], |\mathcal{V}|=t\}$, i.e., ${[b]\choose t}$ is the collection of all $t$-sized subsets of $[b]$. We use $a|b$ to denote that $b$ is divisible by $a$.

\section{Preliminaries}\label{sec_prob}
{\color{black}In this section, we introduce the concept of  PDA and its relationship to the $F$-division $(K,M,N)$ coded caching scheme. Some useful PDAs and their constructions are presented for further discussion.

We first introduce a data placement array $\mathbf{A}=(a_{j,k})_{j\in[F],k\in[K]}$ that represents the data placement strategy. Specifically, given any $F$-division $(K,M,N)$ coded caching scheme under identically uncoded placement strategy, we can construct an $F\times K$ data placement array $\mathbf{A}=(a_{j,k})_{j\in[F],k\in[K]}$ where each entry $a_{j,k}=*$ if user $k\in[K]$ caches the $j$-th packet of all the files, otherwise $a_{j,k}=Null$.  In the following, we introduce PDA \cite{YCTC} that is closely related to data placement array $\mathbf{A}=(a_{j,k})_{j\in[F],k\in[K]}$, but can characterize both the data placement and delivery strategies.}

\subsection{Placement Delivery Array}
\begin{definition}(\cite{YCTC})\rm
\label{def-PDA}
For positive integers $K$,  $F$, and $S$, an $F\times K$ array $\mathbf{P}=(p_{j,k})$, $j\in [F], k\in[K]$,
composed of a specific symbol $``*"$ called star and $S$ symbols in $[S]$, is called a $(K,F,Z,S)$ PDA if it satisfies the following conditions.
\begin{itemize}
\item[C$1$.] Each column has exactly $Z$ stars;
\item [C$2$.] Each integer occurs in $\mathbf{P}$ at least once;
\item [C$3$.] For any two distinct entries $p_{j_1,k_1}$ and $p_{j_2,k_2}$, $p_{j_1,k_1}=p_{j_2,k_2}=s$ is an integer  only if
\begin{enumerate}
	\item [a.] $j_1\ne j_2$, $k_1\ne k_2$, i.e., they lie in distinct rows and distinct columns; and
	\item [b.] $p_{j_1,k_2}=p_{j_2,k_1}=*$, i.e., the corresponding $2\times 2$  subarray formed by rows $j_1,j_2$ and columns $k_1,k_2$ must be of the following form
	\begin{eqnarray*}
		\left(\begin{array}{cc}
			s & *\\
			* & s
		\end{array}\right)~\textrm{or}~
		\left(\begin{array}{cc}
			* & s\\
			s & *
		\end{array}\right).
	\end{eqnarray*}
\end{enumerate}
\end{itemize}
\end{definition}
For instance, it is easy to verify that the array is a $(6,4,2,4)$ PDA.
\begin{eqnarray}
\label{eq-E-pda-1}
\mathbf{P}=\left(\begin{array}{cccccc}
*&*&*&1&2&3\\
*&1&2&*&*&4\\
1&*&3&*&4&*\\
2&3&*&4&*&*
\end{array}\right).
\end{eqnarray}

Yan et al. in \cite{YCTC} showed that any PDA can be used to generate a coded caching scheme, leading to the following result.
\begin{lemma}(\cite{YCTC})\rm
\label{le-Fundamental} If there exists a $(K,F,Z,S)$ PDA, then by Algorithm \ref{alg:PDA} we can obtain an $F$-division $(K,M,N)$ coded caching scheme with $\frac{M}{N}=\frac{Z}{F}$ and transmission load $R=\frac{S}{F}$.
\end{lemma}
\begin{algorithm}[htb]
\caption{Coded caching scheme based on PDA in \cite{YCTC}}\label{alg:PDA}
\begin{algorithmic}[1]
\Procedure {Placement}{$\mathbf{P}$, $\mathcal{W}$}
\State Split each file $W_n\in\mathcal{W}$ into $F$ packets, i.e., $W_{n}=\{W_{n,j}\ |\ j\in[F]\}$.
\For{$k\in \mathcal{K}$}
\State $\mathcal{Z}_k\leftarrow\{W_{n,j}\ |\ p_{j,k}=*, \forall~n\in [N]\}$
\EndFor
\EndProcedure
\Procedure{Delivery}{$\mathbf{P}, \mathcal{W},{\bf d}$}
\For{$s=1,\cdots,S$}
\State  Server sends $\bigoplus_{p_{j,k}=s,j\in[F],k\in[K]}W_{d_{k},j}$.
\EndFor
\EndProcedure
\end{algorithmic}
\end{algorithm}
Let us take the following example to illustrate Algorithm \ref{alg:PDA}.
\begin{example}\rm
\label{E-pda}
Using the PDA in \eqref{eq-E-pda-1} and Algorithm \ref{alg:PDA},  we can obtain a $4$-division $(6,3,6)$ coded caching scheme as follows.

\begin{itemize}
\item \textbf{Placement Phase}: From Line 2 of  Algorithm \ref{alg:PDA}, we have $W_n=\{W_{n,1},W_{n,2},W_{n,3},W_{n,4}\}$, $n\in [6]$. Then by Lines 3-5, the contents cached by user $1$, $2$, $\ldots$, $6$ are respectively
\begin{eqnarray*}
	\mathcal{Z}_1=&\left\{W_{n,1},W_{n,2}| n\in[6]\right\},\\
	\mathcal{Z}_2=&\left\{W_{n,1},W_{n,3}| n\in[6]\right\},\\
	\mathcal{Z}_3=&\left\{W_{n,1},W_{n,4}| n\in[6]\right\},\\
	\mathcal{Z}_4=&\left\{W_{n,2},W_{n,3}| n\in[6]\right\},\\
	\mathcal{Z}_5=&\left\{W_{n,2},W_{n,4}| n\in[6]\right\},\\
	\mathcal{Z}_6=&\left\{W_{n,3},W_{n,4}| n\in[6]\right\}. 
\end{eqnarray*}

\item \textbf{Delivery Phase}: Assume that the request vector is $\mathbf{d}=(1,2,3,4,5,6)$. By Lines 8-10, the signals sent by the server in four time slots are as follows: $W_{1,3}\oplus W_{2,2}\oplus W_{4,1}$, $W_{1,4}\oplus W_{3,2}\oplus W_{5,1}$,  $W_{2,4}\oplus W_{3,3}\oplus W_{6,1}$, and $W_{4,4}\oplus W_{5,3}\oplus W_{6,2}$.

\end{itemize}
\end{example}

\subsection{Relations Between Data Placement  Array $\mathbf{A}$ and PDA $\mathbf{P}$}
It is not difficult to check that each column in $\mathbf{A}$ has exactly $Z=FM/N$ stars. In the delivery phase, assume the server broadcasts a coded packet at time slot $s$, say $W_{d_{k_1},j_1}\oplus W_{d_{k_2},j_2}\oplus\cdots\oplus W_{d_{k_g},j_g}$ for some integer $g\in[K]$. When the server uses the one-shot delivery strategy, each user $k_i$ where $i\in[g]$ can decode its required packet $W_{d_{k_i},j_i}$ only if it caches the other packets $W_{d_{k_{i'}},j_{i'}}$ where $i'\neq i$. This implies that $a_{k_i,j_i}=Null$ and $a_{k_{i'},j_{i'}}=*$. With one-shot delivery, each required packet is transmitted by the server exactly once. So we can put the
entry $a_{k_i,j_i}$ in the unique integer $s$. If the delivery phase contains $S$ times slots, we can check that the obtained array $\mathbf{P}$ is exactly a $(K,F,Z,S)$ PDA. So we have the following result.
\begin{lemma}\rm
\label{le-relationship}	Consider centralized coded caching systems with identically uncoded placement and one-shot delivery. For any positive integers $K$, $F$, $M$ and $N$ satisfying that $Z=FM/N$ is an integer, an $F$-division $(K,M,N)$ coded caching scheme with $S$ time slots in the delivery phase exists if and only if there exists a $(K,F,Z,S)$ PDA.
\end{lemma}
From the above introduction, in order to get the smallest transmission load for any fixed number of users $K$, subpacketization $F$, and memory ratio $M/N=Z/F$, we only need to construct a PDA with $S$ as small as possible. In addition, we replace all the integers in the non-star entries of a PDA, we can obtain a placement array of the scheme realized by this PDA.

Finally, let us introduce the partition PDA \cite{YCTC,TR,SZG}, which can be unified and constructed in the following rule  \cite{CWZW}.
\begin{construction}[\cite{CWZW}]\rm
\label{constr-3}
For any positive integers $q$, $m$ and $t$ with $0<t<m$, let
\begin{align}
\label{eq-consr-3}
\mathcal{F}&=\left\{{\bf f}=(f_{1},f_{2},\ldots,f_{m}, \sum_{i\in[m]} f_i)\ \Big|\  f_1,f_2,\ldots,f_{m}\in [q]\right\},\nonumber\\
\mathcal{K}&=\{(u,v)\ |\ u\in[m+1], v\in[q]\}.
\end{align}	We can obtain a $((m+1)q,q^m, q^{m-1}, (q-1)q^m)$ PDA  $\mathbf{P}=(p_{{\bf f},(u,v)})_{{\bf f}\in \mathcal{F},(u,v)\in \mathcal{K}}$, where  $|\mathcal{F}|\times |\mathcal{K}|=q^m\times (m+1)q$, using the following way:
\begin{eqnarray}
\label{eq-constr-PDA-OA}
p_{{\bf f},(u,v)}=\left\{
\begin{array}{ll}
*&\ \text{if}\ f_u= v\\
({\bf e},n_{\bf e}) & \textrm{otherwise},
\end{array}
\right.
\end{eqnarray}
where ${\bf e}=(e_1,e_{2},\ldots,e_{m})\in[q]^m$ with $e_u=v$ and $e_i=f_i$ for any integer $i\in [m]\setminus\{u\}$.
$n_{\bf e}$ is the  occurrence  order of vector ${\bf e}$ occurring in column $(u,v)$ and starts from $1$.
\end{construction}
\begin{example}\rm
\label{exm-partition}When $m=2$ and $q=3$, from \eqref{eq-consr-3} we have $\mathcal{F}=\{(1,1,2)$, $(2,1,3)$, $(3,1,1)$, $(1,2,3)$, $(2,2,1)$, $(3,2,2)$, $(1,3,1)$, $(2,3,2)$, $(3,3,3)\}$ and 
$\mathcal{K}=\{(1,1)$, $(1,2)$, $(1,3)$, $(2,1)$, $(2,2)$, $(2,3)$, $(3,1)$, $(3,2)$, $(3,3)\}$. By the rule in \eqref{eq-constr-PDA-OA}, the  PDA in \eqref{Eq:matrix1} can be obtained  (on the top of the next page).
\begin{table*}\begin{equation}\label{Eq:matrix1}
\bordermatrix{%
	&(1,1)      &(1,2)      &(1,3)      &(2,1)     &(2,2)     &(2,3)      &(3,1)      &(3,2)        &(3,3)\cr
	(1,1,2)&*          &(2,1,2)    &(3,1,2)    &*         &(1,2,2)   &(1,3,2)    &(1,1,1)    &*            &(1,1,3)               \cr
	(2,1,3)&(1,1,3)    &*          &(3,1,3)    &*         &(2,2,3)   &(2,3,3)    &(2,1,1)    &(2,1,2)      &*                     \cr
	(3,1,1)&(1,1,1)    &(2,1,1)    &*          &*         &(3,2,1)   &(3,3,1)    &*          &(3,1,2)      &(3,1,3)                  \cr
	(1,2,3)&*          &(2,2,3)    &(3,2,3)    &(1,1,3)   &*         &(1,3,3)    &(1,2,1)    &(1,2,2)      &*                      \cr
	(2,2,1)&(1,2,1)    &*          &(3,2,1)    &(2,1,1)   &*         &(2,3,1)    &*          &(2,2,2)      &(2,2,3)                \cr
	(3,2,2)&(1,2,2)    &(2,2,2)    &*          &(3,1,2)   &*         &(3,3,2)    &(3,2,1)    &*            &(3,2,3)              \cr
	(1,3,1)&*          &(2,3,1)    &(3,3,1)    &(1,1,1)   &(1,2,1)   &*          &*          &(1,3,2)      &(1,3,3)                   \cr
	(2,3,2)&(1,3,2)    &*          &(3,3,2)    &(2,1,2)   &(2,2,2)   &*          &(2,3,1)    &*            &(2,3,3)            \cr
	(3,3,3)&(1,3,3)    &(2,3,3)    &*          &(3,1,3)   &(3,2,3)   &*          &(3,3,1)    &(3,3,2)      &*                   }=\mathbf{P}.
	\end{equation}
\end{table*}
\end{example}

\section{A Lower bound on ${R}^*(F)$}
\label{lower-set}
In this section, we present a novel lower bound on $R^*(F)$ under identically uncoded placement and one-shot delivery, and demonstrate that the bound is tight in many cases. Given a data placement array $\mathbf{A}$, we  define the subsets $\mathcal{A}_1$, $\mathcal{A}_2$, $\ldots$, $\mathcal{A}_{K}$ as
follows.
\begin{eqnarray}
\label{eq-colum-set}
\mathcal{A}_k=\{j\ |\ a_{j,k}\neq *, j\in [F] \}.
\end{eqnarray}
Here $\mathcal{A}_k$ denotes the index set of subfiles that are not cached by user $k\in[K]$.
From the above notation, the following statement holds.
\begin{theorem}
\label{th-PDA-set-bound}
Given an $F$-division $(K,M,N)$ centralized coded caching system with data placement array $\mathbf{A}$ and subpacketization level $F$, the optimal transmission load  $R^*(F)$ satisfies
\begin{eqnarray}
\label{eq-max}
{R^*(F)}\geq \max_{ (i_1,\ldots,i_{K})\in \mathcal{I}} \frac{\sum_{h=1}^{K}\Big|\bigcap_{j=1}^{h}\mathcal{A}_{i_j}\Big|\    }{F},
\end{eqnarray}
where $\mathcal{I}$ denotes the set of all permutations of $[K]$, and $\mathcal{A}_{i_j}$ is defined in \eqref{eq-colum-set}.
\end{theorem}
\begin{proof}
{\color{black}Recall from Lemma \ref{le-relationship} that, given an $F$-division $(K,M,N)$ centralized coded caching system with identically uncoded placement and one-shot delivery and subpacketization level $F$, we can obtain a $(K,F,Z,S)$ PDA, where $S$ denotes the total number of time slots and $Z$ indicates that each user caches $Z$ packets of each file.}  According to \eqref{eq-colum-set}, the sets  $\mathcal{A}_1$, $\mathcal{A}_2$, $\ldots$, $\mathcal{A}_{K}$ can be obtained. For any permutation,
say $(i_1,i_2,\ldots,i_{K})\in \mathcal{I}$, we claim that for any two integers $h<h'\in [K]$ and any two distinct
$$x_1\in \bigcap_{j=1}^{h}\mathcal{A}_{i_j},\ \ \ x_2\in \bigcap_{j'=1}^{h'}\mathcal{A}_{i_{j'}},$$
$p_{x_1,i_h},p_{x_2,i_{h}}$, $p_{x_1,i_{h'}}$ and $p_{x_2,i_{h'}}$ are integers from \eqref{eq-colum-set}. In addition, these integers satisfy the following two statements.
\begin{itemize}
\item By the first property of C$3$ in Definition \ref{def-PDA}, i.e., each integer occurs in each column and each row at most once, we have $p_{x_1,i_h}\neq p_{x_2,i_{h}}$ and $p_{x_2,i_h}\neq p_{x_2,i_{h'}}$.
\item $p_{x_1,i_{h}}\neq p_{x_2,i_{h'}}$ holds by the second property of C$3$ in Definition \ref{def-PDA}. Otherwise assume that $p_{x_1,i_{h}}= p_{x_2,i_{h'}}$. By the second property of C$3$ in Definition \ref{def-PDA}, we have $p_{x_1,i_{h'}}= p_{x_2,i_{h}}=*$. This is impossible since $x_2\in  \cap_{j'=1}^{h'}\mathcal{A}_{i_{j'}}\subseteq\cap_{j=1}^{h}\mathcal{A}_{i_j}$. This implies that the cell $(x_2,i_h)$ contains an integer.
\end{itemize}

Consequently, the number of distinct elements in the cells $(x_1,i_h)$ with $x_1\in \cap_{j=1}^{h}\mathcal{A}_{i_j}$ equals $$|\mathcal{A}_{i_1}|+|\mathcal{A}_{i_1}\bigcap
\mathcal{A}_{i_2}|+\cdots+|\bigcap_{j=1}^{K}\mathcal{A}_{i_j}|.$$ Clearly $S \geq |\mathcal{A}_{i_1}|+|\mathcal{A}_{i_1}\cap \mathcal{A}_{i_2}|+\cdots+|\cap_{j=1}^{K}\mathcal{A}_{i_j}|$.  Then the proof is complete.
\end{proof}

\begin{remark}
\label{remark-optimal} Theorem \ref{th-PDA-set-bound} applies to all schemes that can be represented via appropriate PDAs and can recover optimality results in \cite{WTP,YMA2018,BE} (see the full proof in Appendix \ref{appendix-MN}).  Note that the related works in  \cite{WTP,YMA2018,BE} only established optimal lower bounds of communication load for \emph{specific} subpacketization levels, failing to reveal the \emph{general}  trade-off between the transmission load and subpacketization.  Specifically,   \cite{WTP,YMA2018} proved the optimal load $R^*(F)=K(1-M/N)\frac{1}{KM/N}$ under $F\geq {K\choose KM/N}$ and arbitrary storage $M/N$ (i.e., the MN PDA) and \cite{BE} established the optimal $R^*(F)=h{m\choose a+b}/{m\choose b}$ when $F={m\choose b}$ and $M/N = \frac{{m\choose b}-{m-a\choose b}}{{m\choose b}}$(i.e., the strong edge-colored bipartite graph PDA with $\lambda=0$ in \cite{YTCC}).  In addition, our theorem can establish new optimality results as stated in Theorem 2, where $R^*(F)\approx q-1$ with $F=q^m$ and $M/N=1/q$ for any positive integers $m$ and $q\geq 2$.
\end{remark}

\begin{remark}
Our established theoretical trade-off between subpacketization and communication load offers essential design guidelines for practical application. Specifically, to meet a specified communication overhead, our trade-off identifies the minimum subpacketization level and corresponding computational complexity, thereby guiding the selection of necessary storage, computation, and transmission schemes. Moreover, Eq. \eqref{eq-max} indicates that a communication-efficient, one-shot coded caching scheme should maximize the number of subfiles shared among users to increase the coded multicast gain.

\end{remark}

Let us consider $K=6$, $F=4$ and $Z=1$ to illustrate the inequality in \eqref{eq-max}. It is easy to check that the following array is a $(6,4,1,11)$ PDA.
\begin{eqnarray*}
\mathbf{P}_{4\times 6}=\left(\begin{array}{cccccc}
1    &2     &  3    &* & 7     & 8\\
4    &5     &* &3      & 9     & 10\\
6    &*&5      &2      &11      &*\\
*&6     &4      &1      &* & 11
\end{array}\right).
\end{eqnarray*}
According to \eqref{eq-colum-set}, we have $\mathcal{A}_1=\{1,2,3\}$, $\mathcal{A}_2=\{1,2,4\}$, $\mathcal{A}_3=\{1,3,4\}$, $\mathcal{A}_4=\{2,3,4\}$, $\mathcal{A}_5=\{1,2,3\}$, and
$\mathcal{A}_6=\{1,2,4\}$. By applying \eqref{eq-max}, we have
\begin{align*}
&\max\left\{\sum_{h=1}^{6}\Big|\bigcap_{j=1}^{h}\mathcal{A}_{i_j}\Big|\ \Big|\ (i_1,\ldots,i_{6})\in \mathcal{I}\right\}\\
=&|\mathcal{A}_1|+|\mathcal{A}_1\bigcap \mathcal{A}_5|+|\mathcal{A}_1\bigcap \mathcal{A}_5\bigcap \mathcal{A}_2|+|\mathcal{A}_1\bigcap \mathcal{A}_5\bigcap\mathcal{A}_2\bigcap \mathcal{A}_6|+|\mathcal{A}_1\bigcap \mathcal{A}_5\bigcap \mathcal{A}_2\bigcap \mathcal{A}_6\bigcap \mathcal{A}_3|\\
=&3+3+2+2+1\\
=&11.
\end{align*}

For any scheme with identically uncoded placement, one-shot delivery, and subpacketization level $F$, by Theorem \ref{th-PDA-set-bound} we should first find one of the best sorting orders for all the subsets $\mathcal{A}_1$, $\mathcal{A}_2$, $\ldots$, $\mathcal{A}_{K }$ of $[F]$ and then try to compute the exact value of $\sum_{h=1}^{K}|\cap_{j=1}^{h}\mathcal{A}_{i_j}|$ for each integer $h\in [2:K]$. This implies that we should carefully consider both the combinatorial structure of each placement strategy and compute the size of union. With the lower bound in Theorem \ref{th-PDA-set-bound}, by exploiting the combinatorial structures of the partition PDA and enumerative combinatorics,  we prove that the existing schemes in \cite{YCTC,SZG,CWZW,CJTY} achieve the optimal rate-subpacketization trade-off $R^*(F)$ under their respective placement strategies. The results are as follows, with their proofs provided in Section \ref{sec-optimal-2}.
\begin{theorem}\rm\label{th-second-optimality}
For any positive integers $m$ and $q\geq 2$, the $((m+1)q,q^m, q^{m-1}, (q-1)q^m)$ Partition PDA in \cite{YCTC} achieves its optimal rate-subpacketization trade-off if $q=2$ and the approximately optimal rate-subpacketization trade-off if $q>2$.
\end{theorem}

Finally, Theorem \ref{th-PDA-set-bound} can also be generalized to yield the lower bound on the transmission load for schemes with any placement strategy and one-shot delivery, which is formally stated as the following Theorem.

\begin{theorem}\rm
\label{th-PDA-set-lower bound}
Given any one-shot $F$-division $(K,M,N)$ centralized coded caching system with the general identically uncoded placement and subpacketization level $F$, the optimal transmission load $R^*(F)$ satisfies
\begin{align}
	R^*(F)\geq \min_{\mathcal{A}_k\in \Omega,k\in[K]}     \max_{ (i_1,\ldots,i_{K})\in \mathcal{I}} \frac{\sum_{h=1}^{K}\Big|\bigcap_{j=1}^{h}\mathcal{A}_{i_j}\Big| }{F},\label{eq-lower-bound}
\end{align}where $\Omega$ is the set consisting of all the $(F -Z)$-subsets of
$[F]$,  $\mathcal{I}$ is the permutation set of $[K]$, and $\mathcal{A}_{i_j}$ is defined in \eqref{eq-colum-set}.
\end{theorem}

{\color{black}Now let us take some examples to show that our bound is tight for some cases.
\begin{example}\rm
	\label{exam-th-4}
	When $K=4$, $F=6$ and $Z=3$, through exhaustive search over all  $3$-subsets of
	$\{1,2,\ldots,6\}$ by computer we have 
	$$\mathcal{A}_1=\{4,5,6\},\ \ \mathcal{A}_2=\{2,3,6\},\ \ \mathcal{A}_3=\{1,3,5\},\ \ \mathcal{A}_4=\{1,2,4\}$$
	and a sequence of sorting subsets $\mathcal{A}_1$, $\mathcal{A}_2$, $\mathcal{A}_3$ and $\mathcal{A}_4$. Then we have 
	\begin{align*}
		R^*(F)\geq& \min_{\mathcal{A}_k\in \Omega,k\in[4]}     \max_{ (i_1,\ldots,i_{4})\in \mathcal{I}} \frac{\sum_{h=1}^{4}\Big|\bigcap_{j=1}^{h}\mathcal{A}_{i_j}\Big|\    }{F}\\
		=&\frac{|\mathcal{A}_1|+|\mathcal{A}_1\bigcap\mathcal{A}_2|+|\mathcal{A}_1\bigcap\mathcal{A}_2\bigcap\mathcal{A}_3|+
			|\mathcal{A}_1\bigcap\mathcal{A}_2\bigcap\mathcal{A}_3\bigcap\mathcal{A}_4|}{6}\\
		=&\frac{4}{6}=\frac{2}{3}.
	\end{align*}We can construct the following optimal $(4,6,3,4)$ PDA which is exactly the MN PDA. 
	\begin{eqnarray*}
		\mathbf{P}_{6\times 4}=\left(\begin{array}{cccc}
			*    &*     &1    &2\\
			*    &1     &*    &3\\
			*    &2     &3    &*\\
			1    &*     &*    &4\\
			2    &*     &4    &*\\
			3    &4     &*    &*
		\end{array}\right).
	\end{eqnarray*} 
\end{example}

\begin{example}\rm  When $F=8$, $K=6$ and $Z=5$, through exhaustive search over all  $3$-subsets of 
	$\{1,2,\ldots,8\}$  by computer we have $\mathcal{A}_1=\{1,2,4\}$, $\mathcal{A}_2=\{2,3,5\}$,  $\mathcal{A}_3=\{3,4,6\}$, $\mathcal{A}_4=\{5,6,8\}$,  $\mathcal{A}_5=\{6,7,1\}$, $\mathcal{A}_6=\{6,8,2\}$ and a sequence of sorting subsets $\mathcal{A}_1$, $\mathcal{A}_2$, $\mathcal{A}_3$,  $\mathcal{A}_4$, $\mathcal{A}_5$ and $\mathcal{A}_6$. 
	Then we have 
	\begin{align*}
		R^*(F)\geq& \min_{\mathcal{A}_k\in \Omega,k\in[6]}     \max_{ (i_1,\ldots,i_{6})\in \mathcal{I}} \frac{\sum_{h=1}^{6}\Big|\bigcap_{j=1}^{h}\mathcal{A}_{i_j}\Big|\    }{F}\\
		=&\frac{|\mathcal{A}_1|+|\mathcal{A}_1\bigcap\mathcal{A}_2|+
			\cdots+|\mathcal{A}_1\bigcap\mathcal{A}_2\bigcap\mathcal{A}_3\bigcap\mathcal{A}_4\bigcap\mathcal{A}_5\bigcap\mathcal{A}_6|}{8}\\
		=&\frac{4}{8}=\frac{1}{2}.
	\end{align*}We can construct the following optimal $(6,8,5,5)$ PDA.   
	\begin{eqnarray*}
		\mathbf{P}_{8\times 6}=\left(\begin{array}{cccccc}
			1    &*     &*    &*   &4   &*\\
			2    &4     &*    &*   &*   &5\\
			*    &1     &2    &*   &*   &*\\
			3    &*     &4    &*   &*   &*\\
			*    &3     &*    &2   &*   &*\\
			*    &*     &5    &1   &3   &*\\
			*    &*     &*    &*   &2   &1\\
			*    &*     &*    &4   &*   &3
		\end{array}\right).
	\end{eqnarray*} 
\end{example}

}

Since for any two subsets $A$, $B$ of $[F]$, $\overline{\overline{A}\cup \overline{B}}=A\cap B$ holds,
Theorem \ref{th-PDA-set-lower bound} can also be written in the following way.

\begin{corollary}\rm
\label{co-PDA-set-lower bound1}
Given an $F$-division $(K,M,N)$ centralized coded caching system with identically uncoded placement and one-shot delivery and subpacketization level $F$, the optimal transmission load  $R^*(F)$ under any subpacketization level $F$ satisfies
\begin{align*}
	R^*(F)\geq K - \max_{\mathcal{A}_k\in \Omega,k\in[K]} \min_{(i_0,\ldots,i_{K\!-\!1})\in \mathcal{I} } \frac{\sum_{h=0}^{K-1}|\bigcup_{j=0}^{h}\overline{\mathcal{A}}_{i_j}|}{F}. 
\end{align*}
\end{corollary}

\begin{remark}
The proposed lower-bound framework can be extended to derive transmission load lower bounds for both D2D and combination network settings. It should be noted, however, that a direct application of the current bound to these scenarios yields relatively loose lower bounds, primarily because the present analysis does not fully account for the additional combinatorial constraints inherent in D2D and combination networks. These structural constraints substantially increase the complexity of deriving tight lower bounds.
\end{remark}

\section{The proof of Theorem \ref{th-second-optimality}}
\label{sec-optimal-2}  
In Construction \ref{constr-3}, for each $(u,v)\in \mathcal{K}$ from \eqref{eq-constr-PDA-OA} and \eqref{eq-colum-set} we have
\begin{align}
\label{eq-partition-A}
\mathcal{A}_{u,v}=\{{\bf f}=(f_1,f_2,\ldots,f_{m+1})\in \mathcal{F}\ |\ f_u\neq v\}
\end{align}

By Theorem \ref{th-PDA-set-bound}, the key point is to compute the value
\begin{align}
\label{eq-S*}
S^{*}=\max_{ (i_1,\ldots,i_{K})\in \mathcal{I}} \sum_{h=1}^{K}\Big|\bigcap_{j=1}^{h}\mathcal{A}_{i_j}\Big|.
\end{align} Determining the maximum requires considering all possible permutations of the index set $[K]$. However, individually checking these permutations is computationally infeasible. For example, when $q=5$ and $m=3$, the total number of permutations is $|\mathcal{I}|=((3+1)\times 5)!=2\times 10^{18}$, i.e., there are precisely  $2\times 10^{18}$ different permutations of $\{1,2,\ldots,15\}$. Therefore, to obtain a tighter lower bound on the optimal $S^{*}$, we must address the following two critical challenges: 1) Which permutation of the user set $\{(u,v)|u\in[m+1], v\in[1]\}$ is the best one?  2) Given the optimal permutation, say $$((u_1,v_1), (u_2,v_2),\ldots,(u_{(m+1)q},v_{(m+1)q})),$$  how can we explicitly compute the value of the intersections of the sets $\mathcal{A}_{u_i, v_i}$ for each $i= 1, ..., (m+1)q$?

To address these two challenges, we will first demonstrate that the initial $m$ sets can be $\mathcal{A}_{1,v_1}$, $\mathcal{A}_{2,v_2}$, $\ldots$, $\mathcal{A}_{m,v_m}$, where $v_i\in[q]$ for each $i\in[m]$. We then establish that the $(m+1)$th set must be $\mathcal{A}_{m+1,h}$ for some $h\in[q]$. However, it is a difficult task to derive an explicit expression of the total number of different possible integers for any parameters $q$ and $m$, due to the fact that the following sets are chosen according to the parameters $m$ and $q$. Ultimately, we devise an appropriate selection strategy that exploits the algebraic structure of the problem, allowing us to propose an explicit expression for the sum of the sizes of the intersections for any parameters $q$ and $m$.  Numerical comparisons subsequently show that the performance of our chosen permutation is close to that of the optimal permutation. As outlined in the introduction, our proof can be divided into two parts: the selection of the first $m$ sets and the selection of the remaining sets. 

\subsection{Selection of The First $m$ Sets}
\label{sub-Select-first}
The selection of the first $m$ sets is derived from the following investigation. For any integer $u\in[m+1]$, the intersection of any $z$ sets where $z\in[q]$, say $\mathcal{A}_{u,v_1}$, $\mathcal{A}_{u,v_2}$, $\mathcal{A}_{u,v_z}$, is 
\begin{align*}
\{(f_1,f_2,\ldots,f_{m+1})\in \mathcal{F}\ |\ f_{u}\in[q]\setminus\{v_1,v_2,\ldots,v_z\}\}
\end{align*}where $\mathcal{F}$ is defined in $(4)$. Clearly $|\cap_{j\in[z]}\mathcal{A}_{u,v_j}|=(q-z)q^{m-1}$. Since $\mathcal{F}$ forms a $[m+1,m]$ maximum distance separable code, any $m$ coordinates can determine the last coordinate. So, in the following, we only need to consider the first $m$ coordinates. For any $z$ different integers $u_1$, $u_2$, $\ldots$, $u_{z}\in[m]$ where $z\in[m]$, the intersection of any $z$ sets, say $\mathcal{A}_{u_1,v_1}$, $\mathcal{A}_{u_2,v_2}$, $\mathcal{A}_{u_z,v_z}$, is
\begin{align*}
\mathcal{E}_{z}=\{(f_1,f_2,\ldots,f_{m+1})\in \mathcal{F}\ |\ f_{u_j}\in[q]\setminus\{v_j\},j\in[z]\}.
\end{align*} We can compute $\left|\mathcal{E}_{z}\right|=(q-1)^{z}q^{m-z}$. By computing the ratio $|\cap_{j\in[z]}\mathcal{A}_{u,v_j}|/|\mathcal{E}_{z}|=\frac{(q-z)q^{z-1}}{(q-1)^z}$, the following statement can be obtained.  
\begin{proposition}
\label{pro-m-order-set}
For every integer $z$ with $1< z\leq q$, we have $\frac{(q-z)q^{z-1}}{(q-1)^z}<1$.
\hfill $\square$ 
\end{proposition}The detailed proof is included in Appendix $A$. Proposition \ref{pro-m-order-set} tells us that each of the first $m$ sets must come from a unique family $\{\mathcal{A}_{u,1},\mathcal{A}_{u,2},\ldots,\mathcal{A}_{u,q}\}$ where $u\in[m]$. In addition,  the algebra structure of the placement strategy (i.e. eq. \eqref{eq-consr-3}) allows us to arrange the first $m$ sets in the following order: $\mathcal{A}_{1,q}$, $\mathcal{A}_{2,q}$, $\ldots$, $\mathcal{A}_{m,q}$.

So we have the following result.
\begin{align}
S&\geq |\mathcal{A}_{1,q}|+|\mathcal{A}_{1,q}\bigcap\mathcal{A}_{2,q}|+\cdots+|\bigcap_{u\in[m]}\mathcal{A}_{u,q}|\nonumber\\
&= (q-1)q^{m-1}+(q-2)^2q^{m-2}+\cdots+(q-1)^m\nonumber\\
&=\sum_{u=1}^{m}(q-1)^uq^{m-u}.\label{eq-comput-paratition-1}	
\end{align}  For instance,  consider the star placement in \eqref{Eq:matrix1} of Example \ref{exm-partition}. We have
\begin{align*}
\mathcal{A}_{1,1}&=\{(2,1,3),(3,1,1),(2,2,1),(3,2,2),(2,3,2),(3,3,3)\},\\
\mathcal{A}_{1,2}&=\{(1,1,2),(3,1,1),(1,2,3),(3,2,2),(1,3,1),(3,3,3)\},\\
\mathcal{A}_{1,3}&=\{(1,1,2),(2,1,3),(1,2,3),(2,2,1),(1,3,1),(2,3,2)\},\\
\mathcal{A}_{2,1}&=\{(1,2,3),(2,2,1),(3,2,2),(1,3,1),(2,3,2),(3,3,3)\},\\
\mathcal{A}_{2,2}&=\{(1,1,2),(2,1,3),(3,1,1),(1,3,1),(2,3,2),(3,3,3)\},\\
\mathcal{A}_{2,3}&=\{(1,1,2),(2,1,3),(3,1,1),(1,2,3),(2,2,1),(3,2,2)\},\\
\mathcal{A}_{3,1}&=\{(1,1,2),(2,1,3),(1,2,3),(3,2,2),(2,3,2),(3,3,3)\},\\
\mathcal{A}_{3,2}&=\{(2,1,3),(3,1,1),(1,2,3),(2,2,1),(1,3,1),(3,3,3)\},\\
\mathcal{A}_{3,3}&=\{(1,1,2),(3,1,1),(2,2,1),(3,2,2),(1,3,1),(2,3,2)\}. 
\end{align*}
Now consider the vector $(3,3)$ of size $m=2$. First, we choose the subset $\mathcal{A}_{1,3}$ as the first set and have the following result from \eqref{eq-max}.
\begin{align*}
S\geq S_1=|\mathcal{A}_{1,3}|=(3-1)\times 3^{2-1}=6.
\end{align*}Secondly, we choose the subset $\mathcal{A}_{2,2}$ and obtain the following result from \eqref{eq-max}.
\begin{align*}
S\geq S_1+S_2=S_1+\left|\left(\mathcal{A}_{1,3}\bigcap\mathcal{A}_{2,3}\right)\right|=(3-1)\times 3^{2-1}+(3-1)^{2}\times 3^{2-2}=6+4=10. 
\end{align*}

\subsection{Selection of the Remaining Sets}
\label{sub-Select-last}
Now, let us consider the $(m+1)$-th set. For each integer $u\in[m]$ and any integer $v\in[q-1]$, we have the intersection
\begin{align*}
\mathcal{E}_{m}\cap \mathcal{A}_{u,v}=\{(f_1,\ldots,f_{m+1})\in\mathcal{F}\ |\ f_i\in [q-1],  i\in [m]\setminus\{u\}, f_u\in[q-1]\setminus\{v\}\}.
\end{align*}It is easy to check that $|\mathcal{E}_{m}\cap \mathcal{A}_{u,v}|=(q-1)^{m-1}(q-2)$. In addition, for each integer $v\in[q]$, we have the intersection  
\begin{align*}
\mathcal{E}_{m}\bigcap \mathcal{A}_{m+1,v}=\{(f_1,\ldots,f_{m})\in \mathcal{F}\ |\ f_i\in [q-1],  i\in [m],f_{m+1}\in[q]\setminus\{v\}\}.
\end{align*}Since $\left|\mathcal{E}_{m}\right|=(q-1)^m$, there must exist an integer $h\in[q]$ satisfying $|\mathcal{E}_{m}\cap \mathcal{A}_{m+1,h}|\geq \frac{(q-1)^m}{q}$. It is not difficult to check that $(q-2)q^{m-1}<\frac{(q-1)^m}{q}$ always holds. This implies that the $(m+1)$-th set must be chosen from  $\{\mathcal{A}_{m+1,1}, \mathcal{A}_{m+1,2}, \ldots,  \mathcal{A}_{m+1,q}\}$. 

To derive an explicit expression for the total sum of intersection sizes for general parameters $q$ and $m$, we choose the remaining sets from $\{\mathcal{A}_{m+1,1}, \mathcal{A}_{m+1,2}, \ldots,  \mathcal{A}_{m+1,q}\}$ based on the underlying algebraic structure. We emphasize that our approximating computation is close to the exact value $S^{*}$ defined in \eqref{eq-S*}. When $q=3,4,5$ and $m=2,\ldots,8$, we can obtain the ratio $\mu=S^{*}/S$ in Fig. $2$. In Fig. $2$, the $x$-axis represents the value of $m$, and the $y$-axis represents the ratio of our value to the maximum value. The green, blue, and red lines represent the values $q=3,4,5$, respectively. We can see that when $q$ is smaller, the ratio is closer to $1$. Even when $q$ is large, the ratio remains close to $1$ as $m$ increases. 		
\begin{figure}[t]
\centering
\includegraphics[width=0.4\textwidth]{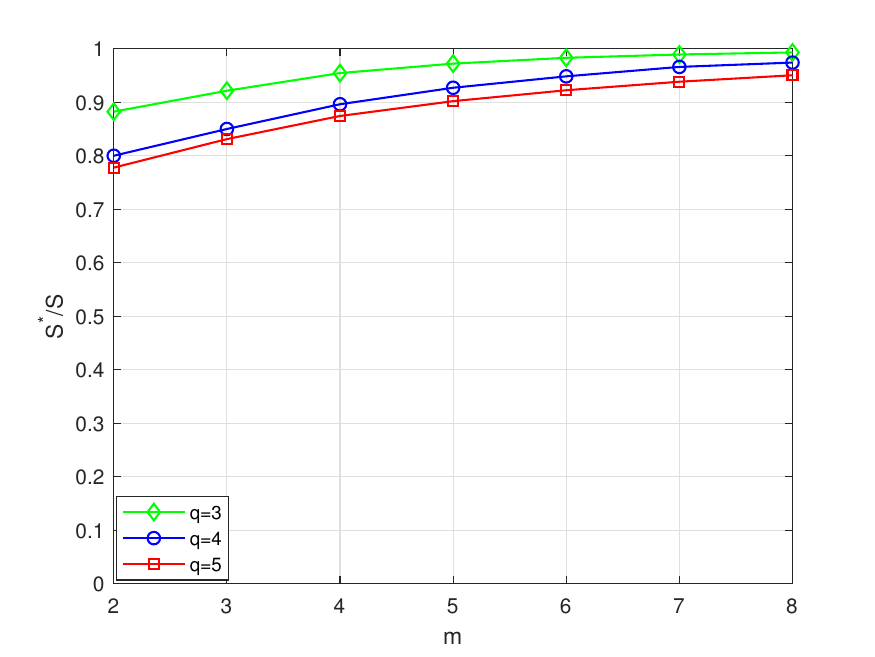}
\caption{The ratios $\mu$ of $S^{*}$ and our derived $S$ with $q=3,4,5$.}
\label{fig-ratio}
\end{figure} 

By the above introduction, we now introduce our detailed computations. Using the algebra structure of $\mathcal{F}$, we will sort the last $q$ subsets $\mathcal{A}_{m+1,q}$,  $\mathcal{A}_{m+1,1}$, $\ldots$,  $\mathcal{A}_{m+l,q-1}$ and take intersections successively.  First, let us take the star placement strategy in \eqref{eq-constr-PDA-OA}	again to introduce our sorting method. From \eqref{eq-comput-paratition-1}, the intersection of  $\mathcal{A}_{1,3}$ and $\mathcal{A}_{2,3}$ is the following set.
\begin{align*}
\mathcal{B}=\{(1,1,2), (2,1,3),(1,2,3),(2,2,1)\}=\{(f_1,f_2,f_1+f_2)|f_1,f_2\in[2]\}.
\end{align*}Using the rows from  $\mathcal{B}$ and the columns derived from the selected subsets $\mathcal{A}_{1,3}$ and $\mathcal{A}_{2,3}$, as well as the remaining subsets $\mathcal{A}_{3,1}$, $\mathcal{A}_{3,2}$ and $\mathcal{A}_{3,3}$, we can obtain the following subarray.
\begin{equation*} 
\bordermatrix{	
&(1,3) &(2,3)      &(3,1)      &(3,2)        &(3,3)\cr
(1,1,2)&13&5&2&*&1 \cr
(2,1,3)&14&12&8&7&*\cr\hline
(1,2,3)&15&6&3&4&*\cr
(2,2,1)&16&11&*&10&9 }=\mathbf{P}'.	
\end{equation*}Clearly we can first choose  $\mathcal{A}_{3,1}$ and then choose $\mathcal{A}_{3,2}$. Then we have
\begin{align*}
S\geq S_1+S_2+\left|\left(\mathcal{A}_{1,3}\bigcap\mathcal{A}_{2,3}\bigcap\mathcal{A}_{3,1}\right)\right|
+\left|\left(\mathcal{A}_{1,3}\bigcap\mathcal{A}_{2,3}\bigcap\mathcal{A}_{3,1}\bigcap\mathcal{A}_{3,2}\right)\right|
=10+3+2=15. 
\end{align*}

Next, let's introduce a key important algebraic property of the $(m+1)$-th coordinate of each vector in $\mathcal{F}$ for the general case. Specifically, the following statement can be directly obtained since $[q]$ forms a cyclic additive group under the addition operation and modulo $q$ operation.
\begin{proposition}\rm
\label{pro-1}
When the values of the $2$nd, $3$rd, $\ldots$, $m$-th coordinates of the vector ${\bf f}$ are fixed, the $(m+1)$-th coordinate will take on $q-1$ different values as the first coordinate varies from $1$ to $q-1$.	\hfill $\square$ 
\end{proposition}According to Proposition \ref{pro-1}, the following notation is useful. For any $q-1$ integers $f_2$, $\ldots$, $f_m\in [q-1]$, define
\begin{align}
\mathcal{F}_{f_2,\ldots,f_m}=\{(i,f_2,f_3,\ldots,f_m, i+\sum_{j\in[2:m]}f_j)|i\in [q-1]\}. 
\label{eq-key-m+1-property}
\end{align}By Proposition \ref{pro-1}, we have $|\mathcal{F}_{f_2,\ldots,f_m}|=q-1$.  For any  $l$ different integers $i_1$, $i_2$, $\ldots$, $i_{l}\in[q]$, define
\begin{align*} 
\mathcal{A}_{m+1,i_1,i_2,\ldots,i_l}=\mathcal{A}_{m+1,i_1}\bigcap \mathcal{A}_{m+1,i_2}\bigcap\cdots\bigcap \mathcal{A}_{m+1,i_l}.
\end{align*}Using the notations of $\mathcal{F}_{f_2,\ldots,f_m}$ and $\mathcal{A}_{m+1,i_1,i_2,\ldots,i_l}$, the following useful result can be obtained.
\begin{lemma}\rm
\label{lemma-key-m+1}
For any $m-1$ integers $f_2$, $\ldots$, $f_m\in [q-1]$ in \eqref{eq-key-m+1-property} and $l$ different integers $i_1$, $i_2$, $\ldots$, $i_{l}\in[q]$ where $l\in[q-1]$, we have
\begin{align*}
\left|\mathcal{A}_{m+1,i_1,i_2,\ldots,i_l}\bigcap\mathcal{F}_{f_2,\ldots,f_m}\right|
=
\left\{
\begin{array}{cc}
	q-l&\ \ \text{if}\ \sum_{j\in[2:m]}f_j=i_{\lambda}, \exists \lambda\in[l]\\
	q-l-1& \text{otherwise}.
\end{array}	
\right.\nonumber
\end{align*}
\end{lemma}
\begin{proof}
From \eqref{eq-partition-A}, we know that a vector ${\bf f}=(f_1$, $f_2$, $\ldots$, $f_{m+1})$ belongs to $\mathcal{A}_{m+1,i_1,i_2,\ldots,i_l}$ if and only $f_{m+1}\in[q]\setminus\{i_1,i_2,\ldots,i_l\}$. This implies that $f_{m+1}$ has exactly $q-l$ possible values. Let $h=\sum_{j\in[2:m]}f_j$. When $h\in \{i_1,i_2,\ldots,i_l\}$, without loss of generality, we assume that $h=i_1$. By the definition of $\mathcal{F}_{f_2,\ldots,f_m}$ in \eqref{eq-key-m+1-property}, we have
$$f_{m+1}=i+\sum_{j\in[2:m]}f_j=h.$$
So we have $i=q$  that contradicts our hypothesis, i.e., $i\neq q$. So $\mathcal{F}_{f_2,\ldots,f_m}\subseteq \mathcal{A}_{m+1,h}$ always holds. Then we have that $	|\mathcal{A}_{m+1,i_1,i_2,\ldots,i_l}\cap\mathcal{F}_{f_2,\ldots,f_m}|=q-l$.

When there is no integer equal to $h$, we have $h\in [q]\setminus\{i_1,i_2,\ldots,i_l\}$. In addition, the integer in the $(m+1)$th coordinate of the vector from $\mathcal{F}_{f_2,\ldots,f_m}$ does not equal $h$. Then the integer in the $(m+1)$-th coordinate of the vector from  $\mathcal{A}_{m+1,i_1,i_2,\ldots,i_l}\cap\mathcal{F}_{f_2,\ldots,f_m}$ can only choose the integer of $[q]\setminus\{i_1,i_2,\ldots,i_l,h\}$. So we have $$\left|\mathcal{A}_{m+1,i_1,i_2,\ldots,i_l}\bigcap\mathcal{F}_{f_2,\ldots,f_m}\right|=q-l-1.$$ Thus the proof is completed.
\end{proof}

Recall that there are $(q-1)^{m-1}$ different $(m-1)$-vectors ${\bf f}'=(f_2,f_3,\ldots,f_m)$ where $f_j\in [q-1]$ for each $j\in [2:m]$, and the sum of each coordinate of these vectors belongs to the set $[q]$. Define
\begin{align*}
\mathcal{E}&=\{{\bf f}=(f_{1},\ldots,f_{m}, \sum_{j\in[m]} f_j) | f_j\in [q-1],j\in[m]\},\\
\mathcal{C}_i&=\{{\bf f}'=(f_{2},\ldots,f_{m})|\sum_{j\in[2:m]}f_j=i, f_j\in [q-1],j\in[m]\}
\end{align*} where $i\in[q]$. 
When $q=3$ and $m=2$ in Example \ref{exm-partition}, we have
$\mathcal{E}=\{(1,1,2),(2,1,3),(1,2,3),(2,2,1)\}$, $\mathcal{C}_1=\{(1)\}$, $\mathcal{C}_2=\{(2)\}$ and $\mathcal{C}_3=\emptyset$. 
Clearly, $\mathcal{E}=\cap_{l=1}^{m}\mathcal{A}_{l,q}$ and the intersection of any two different $\mathcal{C}_i$ and $\mathcal{C}_{i'}$ is empty. In Lemma \ref{lemma-key-m+1}, when $l=1$ and $\sum_{j\in[2:m]}f_j=i$ for any vector ${\bf f}'\in \mathcal{C}_{i}$ where $i\in[q]$, we have $|\mathcal{A}_{m+1,i}\cap\mathcal{F}_{f_2,\ldots,f_m}|=q-1$, i.e., the vector set $\mathcal{F}_{{\bf f}'}\subseteq\mathcal{A}_{m+1,i}$. This implies that $\cup_{{\bf f}'\in\mathcal{C}_{i}}\mathcal{F}_{{\bf f}'}\subseteq \mathcal{A}_{m+1,i}$ for each integer $i\in[q]$ always holds; for any vector ${\bf f}''\in \mathcal{C}_{i'}$ where  $i'\in[q]\setminus\{i\}$, $|\mathcal{F}_{{\bf f}''}\cap\mathcal{A}_{i}|=q-2$ always holds. In addition, the $$\bigcup_{{\bf f}'\in\mathcal{C}_1}\mathcal{F}_{{\bf f}'},\  \bigcup_{{\bf f}'\in\mathcal{C}_2}\mathcal{F}_{{\bf f}'},\  \ldots,\  \bigcup_{{\bf f}'\in\mathcal{C}_q}\mathcal{F}_{{\bf f}'}$$ are exactly the partition of $\mathcal{E}$. So for each integer $i$, the cardinality of $\mathcal{E}\cap\mathcal{A}_{m+1,i}$ is
\begin{align}
\label{eq-number-1}	
|\mathcal{C}_{i}|(q-1)+\sum_{i'\in [q]\setminus\{i\}}|\mathcal{C}_{i'}|(q-2).
\end{align}

From \eqref{eq-number-1}, we will sort the last $q$ subsets $\mathcal{A}_{m+1,q}$,  $\mathcal{A}_{m+1,1}$, $\ldots$,  $\mathcal{A}_{m+l,q-1}$ according to the cardinality of $\mathcal{C}_i$ where $i\in[q]$. According to the value of $m$, we can compute the cardinality of $\mathcal{C}_i$ for each $i\in[q]$. The resulting formula is presented below, with its proof given in Appendix \ref{proof of pro-2}.
\begin{proposition}\rm
\label{pro-2}
\begin{itemize}
\item When $m$ is even, there are exactly $q-1$ different integers $h_i\in [q]$ satisfying that $$|\mathcal{C}_{h_i}|=\left\lceil\frac{(q-1)^{m-1}}{q}\right\rceil=\frac{(q-1)^{m-1}+1}{q}$$ for each $i\in[q-1]$ and exactly one integer $h_q$ satisfying that $$|\mathcal{C}_{h_q}|=\left\lfloor\frac{(q-1)^{m-1}}{q}\right\rfloor=\frac{(q-1)^{m-1}-q+1}{q};$$
\item When $m$ is odd, there are exactly $q-1$ different integers $h_i\in [q]$ satisfying that $$|\mathcal{C}_{h_i}|=\left\lfloor\frac{(q-1)^{m-1}}{q}\right\rfloor=\frac{(q-1)^{m-1}-1}{q}$$ for each $i\in[q-1]$ and exactly one integer $h_q$ satisfying that $$|\mathcal{C}_{h_q}|=\left\lceil\frac{(q-1)^{m-1}}{q}\right\rceil=\frac{(q-1)^{m-1}+q-1}{q}.$$
\end{itemize}
\end{proposition}

When $m$ is even, by Proposition \ref{pro-2} we assume that $$|\mathcal{C}_{i_1}|=|\mathcal{C}_{i_2}|=\cdots=|\mathcal{C}_{i_{q-1}}|=\frac{(q-1)^{m-1}+1}{q}$$ for each $j\in[q-1]$ for some $q-1$ different the integers $i_1$, $i_2$, $\ldots$, $i_{q-1}$. Then we choose $\mathcal{A}_{m+1,i_1}$ as the first subset after our first $m$ selected subsets $\mathcal{A}_{1,q}$,  $\mathcal{A}_{2,q}$, $\ldots$,  $\mathcal{A}_{m,q}$, and have the following result.
\begin{align}
S\geq &\sum_{u=1}^{m}(q-1)^uq^{m-u}+\left|\mathcal{E}\bigcap\mathcal{A}_{m+1,i_1}\right|\label{eq-comput-paratition-2}\\
=&\sum_{u=1}^{m}(q-1)^uq^{m-u}+(q-1)^{m-1}(q-2)\nonumber\\
&+ \frac{(q-1)^{m-1}+1}{q}.\nonumber
\end{align}
By Lemma \ref{lemma-key-m+1}, we choose $\mathcal{A}_{m+1,i_2}$ as the second subset. Let us consider the case $l=2$ in Lemma  \ref{lemma-key-m+1}. If $\sum_{j\in[2:m]}f_j\in\{i_1,i_2\}=i_2$. For any vector ${\bf c}=(f_2,\ldots,f_{m})\in \mathcal{C}_{i_2}$, we have $$|\mathcal{A}_{m+1,i_1,i_2}\bigcap\mathcal{F}_{f_2,\ldots,f_m}|=q-2.$$ 
By the proof of Lemma \ref{lemma-key-m+1}, we have the vector set 
$$(\mathcal{F}_{{\bf f}'}\setminus\{(i_1-i_2,f_2,\ldots,f_m,i_1)\})\subseteq\mathcal{A}_{m+1,i_1,i_2}.$$ 
This implies that 
$$\left(\bigcup_{{\bf c}\in\mathcal{C}_{i_2}}\mathcal{F}_{{\bf c}}\setminus\{(i_1-i_2,f_2,\ldots,f_m,i_1)\}\right)\subseteq \mathcal{A}_{m+1,i_1,i_2}$$ always holds; 
for any vector $${\bf c}'=(f'_2,\ldots,f'_m)\in \mathcal{C}_{i'}$$ where  $i'\not\in\{i_1,i_2\}$, by Lemma \ref{lemma-key-m+1} we have $$|\mathcal{F}_{{\bf c}'}\bigcap\mathcal{A}_{m+1,i_1,i_2}|=q-3.$$ By the proof of Lemma \ref{lemma-key-m+1}, we have 
\begin{align*}
(\mathcal{F}_{{\bf f}''}\setminus\{(i_1-i',f'_2,\ldots,f'_m,i_1),(i_2-i',f_2,\ldots,f_m,i_2)\})\subseteq\mathcal{A}_{m+1,i_1,i_2}.
\end{align*}This implies that 
\begin{align*}
\left(\bigcup_{{\bf c}'\in\mathcal{C}_{i'}}\mathcal{F}_{{\bf c}'}\setminus\{(i_1-i',f_2,\ldots,f_m,i_1),(i_2-i',f_2,\ldots,f_m,i_2)\}\right)
\subseteq \mathcal{A}_{m+1,i_1,i_2}
\end{align*} always holds.  In addition, the 
\begin{align*}
\bigcup_{{\bf f}'\in\mathcal{C}_1}\mathcal{F}_{{\bf f}'}\bigcap\mathcal{A}_{m+1,i_1},\ \bigcup_{{\bf f}'\in\mathcal{C}_2}\mathcal{F}_{{\bf f}'}\bigcap\mathcal{A}_{m+1,i_1}, \ldots,\  \bigcup_{{\bf f}'\in\mathcal{C}_q}\mathcal{F}_{{\bf f}'}\bigcap\mathcal{A}_{m+1,i_1}
\end{align*} are exactly the partition of $\mathcal{E}\cap\mathcal{A}_{m+1,i_1}$. So for each integer $i_1$, the cardinality of $\mathcal{E}\cap\mathcal{A}_{m+1,i_1}$ is
\begin{align}
\label{eq-number-2}	
|\mathcal{C}_{i_1}|(q-2)+|\mathcal{C}_{i_2}|(q-2)+\sum_{i'\in [q]\setminus\{i_1,i_2\}}|\mathcal{C}_{i'}|(q-3).
\end{align} So after our first $m+2$ selected subsets $\mathcal{A}_{1,q}$,  $\mathcal{A}_{2,q}$, $\ldots$,  $\mathcal{A}_{m,q}$, $\mathcal{A}_{m+1,i_1}$, $\mathcal{A}_{m+1,i_2}$, we have the following result.
\begin{align}
S\geq& \sum_{u=1}^{m}(q-1)^uq^{m-u}+(q-1)^{m-1}(q-2)+ \frac{(q-1)^{m-1}+1}{q}+
\left|\mathcal{E}\bigcap\mathcal{A}_{m+1,h_1}\bigcap\mathcal{A}_{m+1,h_2}\right|\nonumber\\
=& \sum_{u=1}^{m}(q-1)^uq^{m-u}+(q-1)^{m-1}(q-2)+ \frac{(q-1)^{m-1}+1}{q}+2\times\frac{(q-1)^{m-1}+1}{q}(q-2)\nonumber\\
&+\left((q-1)^{m-1} -2\times\frac{(q-1)^{m-1}+1}{q}\right)\times (q-3)\nonumber\\
=&\sum_{u=1}^{m}(q-1)^uq^{m-u}+(q-1)^{m-1}(q-2)+(q-1)^{m-1}(q-3)+ 3\times \frac{(q-1)^{m-1}+1}{q}.
\label{eq-comput-paratition-3}
\end{align}
Similar to the above introduction, we have
\begin{align}
S\geq& \sum_{u=1}^{m}(q-1)^uq^{m-u}+|\mathcal{E}\bigcap\mathcal{A}_{m+1,h_1}|+
|\mathcal{E}\bigcap\mathcal{A}_{m+1,i_1}\bigcap\mathcal{A}_{m+1,i_2}|+\cdots+	|\mathcal{E}\bigcap(\bigcap_{l=1}^{q-1}\mathcal{A}_{m+1,i_l})|\nonumber\\
=&\sum_{u=1}^{m}(q-1)^uq^{m-u}+\sum_{i=2}^{q-1}(q-1)^{m-1}(q-i)+(1+2+\cdots+q-1)\times\frac{(q-1)^{m-1}+1}{q} \nonumber\\
=&\sum_{u=1}^{m}(q-1)^uq^{m-u}+\frac{(q-1)^{m+1}}{2}+\frac{q-1}{2}.	\label{eq-comput-paratition-4}
\end{align}Besides, we have 
\begin{align*}
&\sum_{u=1}^{m}(q-1)^uq^{m-u}\\
=&q^{m-1}\left({1\choose 0}q+(-1)^{1}\right)+ q^{m-2}\left({2\choose 0}q^{2}+{2\choose 1}(-1)^{1}q^{1}+(-1)^{2}\right)+q^{m-3}\left({3\choose 0}q^{3}+{3\choose 1}(-1)^{1}q^{2}\right.\\
&\left.+{3\choose 2}(-1)^{2}q^{1}+(-1)^{3}\right)+\cdots+q^{0}\left({m\choose 0}q^{m}+{m\choose 1}(-1)^{1}q^{m-1}+{m\choose 2}(-1)^{2}q^{m-2}+\cdots+(-1)^{m}\right)\\
=&mq^{m}-q^{m-1}\left(1+{2\choose 1}+{3\choose 1}+\cdots+{m\choose 1}\right)+q^{m-2}\left(1+{3\choose 2}+{4\choose 2}+\cdots+{m\choose 2}\right)+\cdots\\
&+(-1)^{m-1}\left({m-1\choose m-1}+{m\choose m-1}\right)q+(-1)^{m}\\
=&mq^{m}-{m+1\choose 2}q^{m-1}+{m+1\choose 3}q^{m-2}+\cdots+(-1)^{m-1}{m+1\choose m}q+(-1)^{m}{m+1\choose m+1}\\
=&-\left(q^{m+1}-{m+1\choose 1}q^{m}+{m+1\choose 2}q^{m-1} -{m+1\choose 3}q^{m-2}+\cdots+(-1)^{m}{m+1\choose m+1}\right)+q^{m+1}-q^{m}\\
=&-(q-1)^{m+1}-q^{m}+q^{m+1}=(q-1)q^{m}-(q-1)^{m+1}.
\end{align*} 
From \eqref{eq-max}, we obtain the lower bound of the partition PDA 
\begin{align}
\label{eq-partition-B-1}
&\max\Big\{\sum_{h\in[K]}|\bigcap_{j\in[h]}\mathcal{A}_{i_j}|\ \Big|\ (i_1,i_2,\ldots,i_{K})\in \mathcal{I}\Big\}\nonumber\\
\geq &(q-1)q^{m}-(q-1)^{m+1}+q(q-1)^m -\frac{(q-1)^m(q-2)}{2}+\frac{q-1}{2}	\nonumber\\
=&(q-1)q^{m}-\frac{(q-1)^{m+1}}{2}+\frac{q-1}{2}.
\end{align}

When $m$ is odd, by Proposition \ref{pro-2} and using the same discussion above, we have the lower bound of the partition PDA as follows.
\begin{align*}
&\max\{\sum_{h\in[K]}|\bigcap_{j\in[h]}\mathcal{A}_{i_j}||(i_1,i_2,\ldots,i_{K})\in \mathcal{I}\}\\
\geq& (q-1)q^{m}-(q-1)^{m+1}+q(q-1)^m -\frac{(q-1)^m(q-2)}{2}+\frac{q-1}{2}\\
=&(q-1)q^{m}-\frac{(q-1)^{m+1}}{2}+\frac{q-1}{q}	
\end{align*} 

By Construction \ref{constr-3} we have an $(m+1)-((m+1)q,q^{m},q^{m-1},q^{m}(q-1))$ partition PDA. The ratio between our bounds in \eqref{eq-partition-B-1}, \eqref{eq-partition-B-2} and the $S$ of partition PDA is 
\begin{align}
\label{eq-partition-B-2}
&\max\Big\{\sum_{h\in[K]}|\bigcap_{j\in[h]}\mathcal{A}_{i_j}|\Big|(i_1,i_2,\ldots,i_{K})\in \mathcal{I}\Big\}\nonumber\\
\geq &(q-1)q^{m}-(q-1)^{m+1}+q(q-1)^m -\frac{(q-1)^m(q-2)}{2}+\frac{q-1}{2}\nonumber\\
=&(q-1)q^{m}-\frac{(q-1)^{m+1}}{2}+\frac{q-1}{q}.
\end{align} 
We can see that when $q=2$, we have
\begin{align*}
&1-\frac{1}{2}\times \left(\frac{q-1}{q}\right)^m+\frac{1}{2q^{m}}\\
=& 1-\frac{1}{2}\times \left(\frac{q-1}{q}\right)^m+\frac{1}{q^{m+1}}\\
=&1-\frac{1}{2^{m+1}}+\frac{1}{2^{m+1}}=1.	
\end{align*}

This implies that when $q=2$ the partition PDA in \cite{YCTC} is optimal. Moreover, when $m$ is large,
\begin{align*}
&1-\frac{1}{2}\times \left(\frac{q-1}{q}\right)^m+\frac{1}{2q^{m}}\\
=&1-\frac{1}{2}\times \left(\frac{q-1}{q}\right)^m+\frac{1}{q^{m+1}}\\
\approx&1.	
\end{align*}This implies that when $q>2$,  the partition PDA in \cite{YCTC} is asymptotically optimal.

\section{Conclusion}
\label{conclusion}
In this paper, we investigated the fundamental trade-off between the subpacketization level and transmission load, and proposed a lower bound on optimal transmission load for any centralized coded caching system with any subpacketization level $F$ under an identically uncoded placement strategy and one-shot delivery strategy. Our work is the first attempt in the literature to theoretically establish the trade-off between subpacketization and transmission load. By carefully exploiting the combinatorial structure and computing the union size of sorted sets, we established a new optimality result, i.e., the partition scheme can achieve the optimal rate-subpacketization trade-off. However, we acknowledge that the lower bound may not be universally applicable for all parameter combinations, suggesting the need for further research to improve it. Additionally, extending this bound to more general scenarios, such as decentralized systems and multi-shot delivery cases, remains a compelling direction for future work.

\section{Acknowledgement}
The authors would like to thank Prof. Xiaohu Tang for his valuable comments and constructive suggestions in improving the presentation and quality of this paper.



\appendices

\section{Optimality of the bipartite graph PDA}
\label{appendix-MN}
\begin{theorem}\rm
	\label{th-first-optimality}
	For any positive integers $m$, $a$, $b$ and $h$ satisfying $a+b<m$, the $({m\choose a},{m\choose b},{m\choose b}-{m-a\choose b},{m\choose a+b})$ bipartite graph PDA and the $(h{m\choose a},{m\choose b},{m\choose b}-{m-a\choose b},h{m\choose a+b})$ grouping bipartite graph PDA  can achieve their optimal rate-subpacketization trade-offs.
\end{theorem} 
\label{proof of the-4} 
\begin{construction}(\cite{MN} Bipartite graph PDA)\rm
	\label{constr-2} For any positive integers $m$, $b$ and $r$ satisfying $a+b<m$, define $\mathcal{F}={[m]\choose b}$ and $\mathcal{K}={[m]\choose a}$. Then we can obtain an $({m\choose a},{m\choose b},{m\choose b}-{m-a\choose b},{m\choose a+b})$ PDA $\mathbf{P}=(p_{\mathcal{B},\mathcal{C}})_{\mathcal{B}\in \mathcal{F}, C\in \mathcal{K}}$.  where
	\begin{eqnarray}
		\label{eq-rule}
		p_{\mathcal{B},\mathcal{C}}=\left\{
		\begin{array}{cc}
			\mathcal{B}\bigcup \mathcal{C}& \hbox{if} \ \mathcal{B}\bigcap \mathcal{C}=\emptyset,\\
			*&\hbox{Otherwise}.
		\end{array}
		\right.
	\end{eqnarray}
\end{construction}
\begin{construction}(\cite{YTCC} Group bipartite graph PDA)\rm
	\label{constr-GCMN}
	Let $\mathbf{P}'$ be a  $({m\choose a},{m\choose b},{m\choose b}-{m-a\choose b},{m\choose a+b})$ PDA generated by Construction \ref{constr-2}. For any positive integer $h$, it is easy to check that
	\begin{align}
		\label{eq-group-BGPDA}
		&\left(\mathbf{P}',\mathbf{P}'+{F\choose Z+1},\ldots,\mathbf{P}'+h{F\choose Z+1}\right)\\[-0.1cm]
		&\ \ \ \underbrace{\ \ \ \ \ \ \ \ \ \ \ \ \ \ \ \ \ \ \ \ \ \ \ \ \ \ \ \ \ \ \ \ \ \ \ \ \ \ \ \ \ \ \ \ \ \ \ }_{h}\nonumber
	\end{align}
	is an $(h{m\choose a},{m\choose b},{m\choose b}-{m-a\choose b},h{m\choose a+b})$ PDA.
\end{construction}
We use the following example to further illustrate the PDA in Construction \ref{constr-2}.
\begin{example}\label{ex-scbg}\rm When $m=5$, $a=2$ and $b=1$, we have
	\begin{eqnarray*}
		\mathcal{F}&=&\{\{1\},\{2\},\{3\},\{4\},\{5\}\}\\
		\mathcal{K}&=&\{\{1,2\}, \{1,3\},  \{1,4\},\{1,5\},\{2,3\},\{2,4\},\{2,5\},\{3,4\},\{3,5\},\{4,5\}\}.
	\end{eqnarray*}To simplify notations, we represent each subset as a string for short. For instance, the entry $\{1,2,3\}$ in the third row and the first column in \eqref{eq-exam-1} is written as $123$. We will use this subset representation throughout the paper. By the first rule in \eqref{eq-rule}, the following PDA can be obtained.
	\begin{equation}
		\label{eq-exam-1}\mathbf{P}=\bordermatrix{%
			&12&13&14&15&23&24&25&34&35&45\cr
			&*	&	*	&	*	&	*	&	123	&	124	&	125	&	134	&	135	&	145\cr
			&*	&	123	&	124	&	125	&	*	&	*	&	*	&	234	&	235	&	245\cr
			&123	&	*	&	134	&	135	&	*	&	234	&	235	&	*	&	*	&	345\cr
			&124	&	134	&	*	&	145	&	234	&	*	&	245	&	*	&	345	&	*\cr
			&125	&	135	&	145	&	*	&	235	&	245	&	*	&	345	&	*	&	*}
	\end{equation}
\end{example}

We will first show that the bipartite graph PDA achieves our new bound. Then by Theorem $1$, the optimality of the group bipartite graph PDA can be directly obtained based on the optimality of the bipartite graph PDA.

Now let us consider the bipartite graph PDA. In Construction $1$, for each $B\in \mathcal{K}$, we have
\begin{align}
	\label{eq-scg-A}
	\mathcal{A}_{\mathcal{B}}=\left\{\mathcal{C}\ \Big|\ \mathcal{B}\bigcap \mathcal{C}=\emptyset, \mathcal{C}\in \mathcal{F}\right\}
\end{align}The set $\mathcal{A}_{\mathcal{B}}$ consists of all the $b$-subsets $\mathcal{C}\in \mathcal{F}$ such that each of them contains no element from $\mathcal{B}$. So $|\mathcal{A}_{\mathcal{B}}|={m-a\choose b}$. In addition, for any $\lambda$ different subsets $\mathcal{B}_1$, $\mathcal{B}_2$, $\ldots$, $\mathcal{B}_{\lambda}\in \mathcal{K}$, by the star placement strategy and from \eqref{eq-scg-A}, the intersection of $\mathcal{A}_{\mathcal{B}_1}$,  $\mathcal{A}_{\mathcal{B}_2}$, $\ldots$,  $\mathcal{A}_{\mathcal{B}_{\lambda}}$ has exactly
\begin{align}
	\label{eq-key-scbg}
	{|[m]\backslash(\cup_{i\in [\lambda]}\mathcal{B}_i)|\choose b}
\end{align}elements which are only related to the size of the union of these $\lambda$ sets. {This is a key point for us to sort the sets $\mathcal{A}_{\mathcal{B}}$ where $\mathcal{B}\in \mathcal{K}$ and take intersections successively.} For instance, let us consider the star placement in \eqref{eq-exam-1} of Example \ref{ex-scbg}. We have
\begin{align}\label{eq-scbg-ex-1}
	\mathcal{A}_{12}=\{3,4,5\},\  \mathcal{A}_{13}=\{2,4,5\},\  \mathcal{A}_{14}=\{2,3,5\},\  \mathcal{A}_{15}=\{2,3,4\},\ \mathcal{A}_{23}=\{1,4,5\},\nonumber\\
	\mathcal{A}_{24}=\{1,3,5\},\  \mathcal{A}_{25}=\{1,3,4\},\  \mathcal{A}_{34}=\{1,2,5\},\  \mathcal{A}_{35}=\{1,2,4\},\  \mathcal{A}_{45}=\{1,2,3\}.
\end{align}In order to estimate the maximum value $S$ in \eqref{eq-exam-1}, from \eqref{eq-key-scbg} we consider the problem of ordering the sets in \eqref{eq-scbg-ex-1}. First, we choose $\mathcal{A}_{12}$ as the first set and have
\begin{align}
	\label{eq-ex-1-1}
	{ S\geq S_1=|\mathcal{A}_{12}|}={2\choose 2}{5-2\choose 1}=3.
\end{align}Here the first item ${2\choose 2}$ means that we choose a $2$-subset from $[2]$ and the second item ${5-2\choose 1}$ means that there are ${m-a\choose b}$ possible subsets left, satisfying that the intersection of each of them with $[2]$ is empty.  Then from \eqref{eq-key-scbg}, among all these sets in \eqref{eq-scbg-ex-1}, any one of set $\mathcal{A}_{23}$, $\mathcal{A}_{13}$, $\mathcal{A}_{24}$, $
\mathcal{A}_{14}$, $\mathcal{A}_{15}$, $\mathcal{A}_{25}$ gives the smallest number of elements in the intersection with set $\mathcal{A}_{12}$. To better introduce our idea, we will choose the next subset in the ascending order of their subscripts. So we will choose the subset $\mathcal{A}_{13}$ and have
\begin{align}
	\label{eq-ex-1-2}
	S\geq S_1+|\left(\mathcal{A}_{12}\cap\mathcal{A}_{13}\right)|={2\choose 2}{5-2\choose 1}+{5-3\choose 1}=3+2=5.
\end{align} Here the second item ${5-3\choose 1}$ means that there are exactly ${5-3\choose 1}$ possible subsets left satisfying that each of them does not contain integer $2$ and the intersection of each of them with $[3]$ is empty. According to the rules we choose for the sets, we should choose the subset $\mathcal{A}_{23}$ next. We have
\begin{align}
	\label{eq-ex-1-3}
	S\geq& S_1+S_2={|\mathcal{A}_{12}|}+\left|\left(\mathcal{A}_{12}\bigcap\mathcal{A}_{13}\right)\right|+\left|\left(\mathcal{A}_{12}\bigcap\left(\mathcal{A}_{13}\bigcap\mathcal{A}_{23}\right)\right)\right|\nonumber\\
	=&{2\choose 2}{5-2\choose 1}+\left({3\choose 1}-1\right){5-3\choose 1}=3+2\times 2=7.
\end{align}Here in the item $\left({3\choose 1}-1\right){5-3\choose 1}$, the first item ${3\choose 1}-1$  means that there are exactly ${3\choose 1}-1$ possible $2$-subset of $[3]$ expect for the $2$-subset $[2]$ which was chosen in the first step. The second item ${5-3\choose 1}$ means that there are exactly ${5-3\choose 1}$ possible subsets left satisfying that the intersection of each of them with $[3]$ is empty. It is interesting that for the subset of $[3]$, in the second step if we first choose the subset  $\mathcal{A}_{23}$ and then choose the subset  $\mathcal{A}_{13}$, we have
$$\left|\left(\mathcal{A}_{12}\bigcap\mathcal{A}_{13}\right)\right|+\left|\left(\mathcal{A}_{12}\bigcap\left(\mathcal{A}_{13}\bigcap\mathcal{A}_{23}\right)\right)\right|=
\left|\left(\mathcal{A}_{12}\bigcap\mathcal{A}_{23}\right)\right|+\left|\left(\mathcal{A}_{12}\bigcap\left(\mathcal{A}_{23}\bigcap\mathcal{A}_{13}\right)\right)\right|.$$
The reason is that $\{1,2\}\cup\{2,3\}=\{1,2\}\cup\{1,3\}=[3]$. So if there are several left subsets satisfying the following condition: the union of each subset with the previously selected subsets is the same, then we do not necessarily have to select them in strictly ascending order of their indices. Instead, we can choose any one of them. For simplicity, in the second step, we first choose the subset  $\mathcal{A}_{23}$ and then choose the subset  $\mathcal{A}_{13}$. By adjusting the subsets $\mathcal{A}_{23}$ and $\mathcal{A}_{13}$ in \eqref{eq-ex-1-2} and \eqref{eq-ex-1-3}, we have
\begin{align}
	\label{eq-ex-1-4}
	S\geq& {S_1}+\left|\left(\mathcal{A}_{12}\bigcap\mathcal{A}_{[3]\setminus\{1\}}\right)\right|={2\choose 2}{5-2\choose 1}+{5-3\choose 1}=3+2=5\nonumber\\
	S\geq&{\color{black}S_1}+S_2=|\mathcal{A}_{12}|+\left(\left|\left(\mathcal{A}_{12}\bigcap\mathcal{A}_{[3]\setminus\{1\}}\right)\right|+
		\left|\bigcap_{i\in[3]}\mathcal{A}_{[3]\setminus\{i\}}\right|\right)\nonumber\\
	=&{2\choose 2}{5-2\choose 1}+\left({3\choose 1}-1\right){5-3\choose 1}=3+2\times 2=7.
\end{align}
Similar to the above selecting steps, we should consider the left subsets $\mathcal{A}_{34}$, $\mathcal{A}_{24}$,
$\mathcal{A}_{14}$ and select them in sequence. Then we have
\begin{align}
	\label{eq-ex-1-5}
	S\geq& S_1+S_2+S_3\nonumber\\
	=&|\mathcal{A}_{12}|+\left(\left|\left(\mathcal{A}_{12}\bigcap\mathcal{A}_{[3]\setminus\{1\}}\right)\right|+
		\left|\bigcap_{i\in[3]}\mathcal{A}_{[3]\setminus\{i\}}\right|\right)+\left|\left(\bigcap_{i\in[3]}\mathcal{A}_{[3]\setminus\{i\}}\right)\bigcap\mathcal{A}_{34}\right|+\nonumber\\
	&\left|\left(\bigcap_{i\in[3]}\mathcal{A}_{[3]\setminus\{i\}}\right)\bigcap\mathcal{A}_{34}\right|
		+
		\left|\left(\bigcap_{i\in[3]}\mathcal{A}_{[3]\setminus\{i\}}\right)\bigcap\mathcal{A}_{34}\bigcap\mathcal{A}_{24}\right|  +\left|\left(\bigcap_{i\in[3]}\mathcal{A}_{[3]\setminus\{i\}}\right)\bigcap_{\mathcal{B}\in {[3]\choose 2}}\mathcal{A}_{[4]\setminus \mathcal{B}}\right|
	\nonumber\\
	=&{2\choose 2}{5-2\choose 1}+\left({3\choose 1}-1\right){5-3\choose 1}+\left({4\choose 2}-{3\choose 2}\right){5-4\choose 1}\nonumber\\
	=&3+2\times 2+(6-3)=10.
\end{align}In term $\left({4\choose 2}-{3\choose 2}\right){5-4\choose 1}$, the number ${4\choose 2}-{3\choose 2}$ means that among the  ${4\choose 2}$ $2$-subset of $[4]$, there are ${3\choose 2}$ $2$-subsets were selected in the previous steps, i.e., there are  ${4\choose 2}-{3\choose 2}$ $2$-subsets left. Similarly, we use the same selecting method to consider the left subsets $\mathcal{A}_{15}$, $\mathcal{A}_{25}$, $\mathcal{A}_{35}$, $\mathcal{A}_{45}$. We can check that the intersection of each of $\mathcal{A}_{15}$, $\mathcal{A}_{25}$, $\mathcal{A}_{35}$, $\mathcal{A}_{45}$ with $\cap_{\mathcal{B}\in {[4]\choose 2}}\mathcal{A}_{[4]\setminus \mathcal{B}}$ is empty. So it is not necessary to consider the left subsets. From the above introduction, we can see that we sort the sets in \eqref{eq-scbg-ex-1} in the following order
\begin{align}\label{eq-scbg-order}
	\mathcal{A}_{12},\ \mathcal{A}_{23},\ \mathcal{A}_{13},\ \mathcal{A}_{34},\ \mathcal{A}_{24},\
	\mathcal{A}_{14},\ \mathcal{A}_{45},\ \mathcal{A}_{15},\ \mathcal{A}_{25},\ \mathcal{A}_{35}.
\end{align}By Theorem $1$, we have the following lower bound for the star placement in \eqref{eq-exam-1}.
$$\max\Big\{\sum_{h\in[K]}|\bigcap_{j\in[h]}\mathcal{A}_{i_j}|\ \Big|\ (i_1,i_2,\ldots,i_{K})\in \mathcal{I}\Big\}\geq S_1+S_2+S_3=10.$$
Clearly, the PDA in \eqref{eq-exam-1} achieves our lower bound. This implies that the PDA in \eqref{eq-exam-1} is optimal under the given star placement, and our lower bound is tight for this case.

According to the above example, for any parameters $m$, $a$, and $b$, we can also sort the sets as in \eqref{eq-ex-1-5} and take intersections successively.  That is,
\begin{subequations}
	\begin{align}
		&\max\Big\{\sum_{h\in[K]}|\bigcap_{j\in[h]}\mathcal{A}_{i_j}|\ \Big|\ (i_1,i_2,\ldots,i_{K})\in \mathcal{I}\Big\}\label{eq-SBG-0}\\
		\geq&|\mathcal{A}_{[a]}|+|\mathcal{A}_{[a]}\bigcap
		\mathcal{A}_{[a+1]\setminus\{1\}}|+
		|\mathcal{A}_{[a]}\bigcap(
		\bigcap_{J_1\in{[2]\choose 1}}\mathcal{A}_{[a+1]\setminus J_1})|+\cdots+
		|\bigcap_{J_1\in{[a+1]\choose 1}}\mathcal{A}_{[a+1]\setminus J_1}|+
		\label{eq-SBG-1}\\
		&|(
		\bigcap_{J_1\in{[a+1]\choose 1}}\mathcal{A}_{[a+1]\setminus J_1})\bigcap
		\mathcal{A}_{[a+2]\setminus\{1,2\}}|+|(
		\bigcap_{J_1\in{[a+1]\choose 1}}\mathcal{A}_{[a+1]\setminus J_1})\bigcap
		\mathcal{A}_{[a+2]\setminus\{1,2\}}\bigcap
		\mathcal{A}_{[a+2]\setminus\{1,3\}}|+\nonumber\\
		&
		|(
		\bigcap_{J_1\in{[a+1]\choose 1}}\mathcal{A}_{[a+1]\setminus J_1})
		\bigcap
		(\bigcap_{J_2\in {[3]\choose 2}}
		\mathcal{A}_{[a+2]\setminus J_2})|+
		\cdots+
		|(
		\bigcap_{J_1\in{[a+1]\choose 1}}\mathcal{A}_{[a+1]\setminus J_1})
		\bigcap
		(\bigcap_{J_2\in {[a+1]\choose 2}}
		\mathcal{A}_{[a+2]\setminus J_2})|+\label{eq-SBG-2}\\
		&
		|(
		\bigcap_{J_2\in{[a+2]\choose 2}}\mathcal{A}_{[a+2]\setminus J_2})
		\bigcap\mathcal{A}_{[a+3]\setminus\{1,2,3\}}|+\cdots+|
		\bigcap_{J_2\in{[a+3]\choose 3}}\mathcal{A}_{[a+3]\setminus J_3}|+
		\cdots+|\bigcap_{J_{3}\in{[m+3]\choose 3}}
		\mathcal{A}_{[m]\setminus J_{3}}|\label{eq-SBG-3}+\\
		&\cdots+
		\cdots+|\bigcap_{J_{m-a}\in{[m]\choose m-a}}
		\mathcal{A}_{[m]\setminus J_{m-a}}|\nonumber\\
		&={a\choose a}{m-a\choose b}+\left({a+1\choose a}-{a\choose a}\right){m-a-1\choose b}+\left({a+2\choose a}-{a+1\choose a}\right){m-a-2\choose b}\nonumber\\
		&+\cdots+\left({a+i\choose a}-{a+i-1\choose a}
		\right){m-a-i\choose b}+\cdots+\left({m-a-b\choose a}-{m-a-b-1\choose a}\right){b\choose b}\nonumber\\
		&={a\choose a}{m-a\choose b}+{a\choose a-1}{m-a-1\choose b}+{a+1\choose a-1}{m-a-2\choose b}+\cdots
		\label{eq-sec-1}\\
		&+{a+i\choose a-1}{m-a-i-1\choose b}+\cdots+{m-b-1\choose a-1}{b\choose b}\nonumber
	\end{align}
\end{subequations}Here we can see that $|\mathcal{A}_{[a]}|$ in \eqref{eq-SBG-1} is the number of non-star entries, each of which contains the integer set $[a]$ and is labeled by column $[a]$.

Next, we show that
\begin{IEEEeqnarray}{rCl}
	&&{a\choose a}{m-a\choose b}+{a\choose a-1}{m-a-1\choose b}+{a+1\choose a-1}{m-a-2\choose b}+\cdots \nonumber
	\\
	&& +{a+i\choose a-1}{m-a-i-1\choose b}+\cdots +{m-b-1\choose a-1}{b\choose b} = {m\choose a+b}. \label{eq21}
\end{IEEEeqnarray}
From \eqref{eq-rule}, we have $|\mathcal{A}_{[a]}|={m-a\choose b}$. We can regard it as the number of $(a+b)$-subsets each of which contains $[a]$; $|\mathcal{A}_{[a]}\bigcap
\mathcal{A}_{[a+1]\setminus\{1\}}|+
|\mathcal{A}_{[a]}\bigcap(
\bigcap_{J_1\in{[2]\choose 1}}\mathcal{A}_{[a+1]\setminus J_1})|+\cdots+
|\bigcap_{J_1\in{[a+1]\choose 1}}\mathcal{A}_{[a+1]\setminus J_1}|$ in \eqref{eq-SBG-1} is the number of non-star entries each of which contains the integer set $[a+1]\setminus\{j\}$ for each and is also column labeled by the set $[a+1]\setminus\{j\}$ for some integer $i\in[a]$;
\begin{align*}
	&\left|\left(
	\bigcap_{J_1\in{[a+1]\choose 1}}\mathcal{A}_{[a+1]\setminus J_1}\right)\bigcap
	\mathcal{A}_{[a+2]\setminus\{1,2\}}\right|+\left|\left(
	\bigcap_{J_1\in{[a+1]\choose 1}}\mathcal{A}_{[a+1]\setminus J_1}\right)\bigcap
	\mathcal{A}_{[a+2]\setminus\{1,2\}}\bigcap
	\mathcal{A}_{[a+2]\setminus\{1,3\}}\right|\\
	&+
	\left|\left(
	\bigcap_{J_1\in{[a+1]\choose 1}}\mathcal{A}_{[a+1]\setminus J_1}\right)
	\bigcap
	\left(\bigcap_{J_2\in {[3]\choose 2}}
	\mathcal{A}_{[a+2]\setminus J_2}\right)\right|+\cdots 
	+
	\left|\left(
	\bigcap_{J_1\in{[a+1]\choose 1}}\mathcal{A}_{[a+1]\setminus J_1}\right)
	\bigcap
	\left(\bigcap_{J_2\in {[a+1]\choose 2}}
	\mathcal{A}_{[a+2]\setminus J_2}\right)\right|
\end{align*} in \eqref{eq-SBG-2} is the number of non-star entries, each of which contains the integer set $[a+1]\setminus\{j\}$ and is also column labeled by the integer set $[a+2]\setminus\{j_1,j_2\}$ for some different integers $j_1,j_2\in [a+1]$ and so on.

From the above introduction, the first item of ${a+i\choose a-1}{m-a-i-1\choose b-1}$ in \eqref{eq-sec-1} can be regarded as the number of $(a+b)$-subset containing integer $a+i$ and $a-1$ integer in $[a+i-1]$. In addition, for each integer $0\leq i\leq m-a-b-1$, each such $(a+b)$-subset $D$ uniquely contains an $a$-subset consisting of the first $a$ smallest integers in $D$. So, by the above introduction,  \eqref{eq-sec-1} is exactly the number of all the $(a+b)$-subsets of $[m]$.

From \eqref{eq-SBG-0} and \eqref{eq21} we have the following result.
\begin{align}
	\max\Big\{\sum_{h\in[K]}|\bigcap_{j\in[h]}\mathcal{A}_{i_j}|\ \Big|\ (i_1,i_2,\ldots,i_{K})\in \mathcal{I}\Big\}
	\geq {m\choose a+b}. \label{eq22}
\end{align}
From the  Construction \ref{constr-2} and \eqref{eq22}, we obtain that the bipartite graph PDA achieves the optimal rate-subpacketization trade-off.

For the sort of all the subsets for the bipartite graph PDA, without loss of generality, we assume the sort is $\mathcal{A}_1$, $\mathcal{A}_2$, $\ldots$, $\mathcal{A}_{{m\choose a}}$. From  \eqref{eq-group-BGPDA}, we can obtain a sort subset for the group bipartite graph PDA  as
\begin{align*}{
		\underbrace{\mathcal{A}'_1=\mathcal{A}_1, \ldots, \mathcal{A}'_h=\mathcal{A}_1}_h,\underbrace{\mathcal{A}'_{h+1}=\mathcal{A}_2, \ldots, \mathcal{A}'_{2h}=\mathcal{A}_2}_h,\ldots,\underbrace{\mathcal{A}'_{({m\choose a}-1)h+1}=\mathcal{A}_{{m\choose a}}, \ldots, \mathcal{A}'_{{m\choose a}h}=\mathcal{A}_{{m\choose a}}}_h.}
\end{align*}Then, by Theorem $1$ we have
\begin{align}
	\max\Big\{\sum_{h\in[K]}|\bigcap_{j\in[h]}\mathcal{A}'_{i_j}|\ \Big|\ (i_1,i_2,\ldots,i_{K})\in \mathcal{I}\Big\}
	\geq h{m\choose a+b}. \label{eq-group}
\end{align}
From the  Construction \ref{constr-GCMN} and \eqref{eq-group}, we also obtain that the group bipartite graph PDA achieves the optimal rate-subpacketization trade-off.

\section{Proof of Proposition~\ref{pro-m-order-set}}
\label{proof-pro-m-order-set} 
For any positive integer $z\leq \min\{q,m\}$, let us consider the ratio function 
\begin{align*}
	\phi(z):=\frac{(q-z)q^{\,z-1}}{(q-1)^z}. 
\end{align*}In the following, we will show that $\phi(z)\le 1$ holds for each $z\in[2:q]$. 
Let us compare consecutive values $\phi(z+1)$ and $\phi(z)$ for integers $z$ with $2\leq z< q-1$. That is,  
\begin{align*}
	\frac{\phi(z+1)}{\phi(z)}
	=&\frac{(q-(z+1))q^{z}/(q-1)^{z+1}}{(q-z)q^{z-1}/(q-1)^z}\\
	=&\frac{q}{q-1}\cdot\frac{q-z-1}{q-z}\\
	=&\left(1+\frac{1}{q-1}\right)\left(1-\frac{1}{q-z}\right)\\
	=&1+\frac{1}{q-1}-\frac{1}{q-z}-\frac{1}{(q-1)(q-z)}\\
	=&1-\frac{z-2}{(q-1)(q-z)}\\
	\leq &1.
\end{align*} This implies that $\phi(z+1)\leq \phi(z)$ holds for any positive integer $z\in[2:q-1]$. When $z=2$, we have $$\phi(2)=\frac{q(q-2)}{(q-1)^2}=1-\frac{1}{(q-1)^2}<1.$$ So $\phi(z)\leq \phi(2)<1$ holds for each positive integer $z\in[3:q-1]$.
Then the proof is completed. 

\section{Proof of Proposition~\ref{pro-2}}
\label{proof of pro-2} 
We will use induction to prove our statement. Let us first consider the case $m=2$. Then we have
$$\mathcal{C}_{1}=\{1\},\ \mathcal{C}_{2}=\{2\},\ \ldots,\ \mathcal{C}_{q-1}=\{q-1\}.$$
Clearly $$|\mathcal{C}_{1}|=|\mathcal{C}_{2}|=\cdots=|\mathcal{C}_{q-1}|=
\frac{(q-1)^{2-1}+1}{q}=1$$ and $$|\mathcal{C}_{q}|=
\frac{(q-1)^{2-1}-q+1}{q}=0.$$ So our statement holds for the case $m=2$. Let $$\mathcal{C}^{(2)}=\mathcal{C}_{1}\bigcup\mathcal{C}_{2}\bigcup \cdots\bigcup\mathcal{C}_{q-1}\bigcup\mathcal{C}_{q}$$ for the case $m=2$. Now let us consider the case $m=3$ by $\mathcal{C}^{(2)}$.

When $m=3$, for all the $f_2, f_3\in[q-1]$ all the values of $f_2+f_3$ can be represented in the following multi-set.
\begin{align}
&f_2+f_3\in\ \left(\mathcal{C}^{(2)}+1\right)\bigcup \left(\mathcal{C}^{(2)}+2\right)\bigcup \cdots\bigcup\left(\mathcal{C}^{(2)}+q-1\right)\nonumber\\
=	&\{2,3,\ldots,q\}
\bigcup\{3,4,\ldots,1\}\bigcup \cdots\bigcup \{q,1,\ldots,q-2\}\label{eq-case-3}\nonumber\\
=&\mathcal{C}^{(3)}.
\end{align}From \eqref{eq-case-3} we can see that each integer $i\in[q-1]$, $\mathcal{C}^{(2)}+i$ does not contain $i$ but contains integer $q$. That is, each integer $i$ where $i\in[q-1]$ occurs in $\mathcal{C}^{(3)}$ exactly $q-2$ times and $q$ occurs in $\mathcal{C}^{(3)}$ exactly $q-1$ times. By the definition of $\mathcal{C}_{i}$, we have $$|\mathcal{C}_{1}|=\cdots=|\mathcal{C}_{q-1}|=
\frac{(q-1)^{3-1}-1}{q}=
\frac{q^2-2q}{q}=q-2$$ and $$|\mathcal{C}_{q}|=
\frac{(q-1)^{3-1}+q-1}{q}=
\frac{q^2-q}{q}=q-1.$$ So our statement holds for the case $m=3$. Now let us consider the case for any integer $m>3$.

Define the multi-set 
\begin{align*}
\mathcal{C}^{(k)}=&\left(\mathcal{C}^{(k-1)}+1\right)\bigcup \left(\mathcal{C}^{(k-1)}+2\right)\bigcup \cdots\bigcup\left(\mathcal{C}^{(k-1)}+q-1\right)
\end{align*} where $k\geq 4$. Assume that the Proposition \ref{pro-2} holds for each $m=k$ where $k>3$. Now let us consider the case $m=k+1$. Similar to the cases $m=2,3$, since the values of the sum $\sum_{j\in[2:m]}f_j$ can be represented in the following multi-set.
\begin{align*}
\sum_{j\in[2:m]}f_j\in &\ \left(\mathcal{C}^{(k)}+1\right)\bigcup \left(\mathcal{C}^{(k)}+2\right)\bigcup \cdots\bigcup\left(\mathcal{C}^{(k)}+q-1\right)\\
=&\mathcal{C}^{(k+1)} 
\end{align*}
If $k$ is even, by our hypothesis there are exactly $q-1$ different integers $h_i\in [q]$ satisfying that $|\mathcal{C}_{h_i}|=((q-1)^{k-1}+1)/q$ for each $i\in[q-1]$ and exactly one integer $h_q$ satisfying that $|\mathcal{C}_{h_q}|=(((q-1)^{k-1}-q+1)/q$. Without loss of generality, we assume that $|\mathcal{C}_i|=((q-1)^{k-1}+1)/q$ where $i\in[q-1]$ and $|\mathcal{C}_{q}|=((q-1)^{k-1}-q+1)/q$. This implies that each integer $i\in [q-1]$ occurs in $\mathcal{C}^{(k)}$ exactly $((q-1)^{k-1}+1)/q$ times and $q$ occurs in $\mathcal{C}^{(k)}$ exactly $((q-1)^{k-1}-q+1)/q$ times. It is easy to check that for any two integers $i,i'\in[q-1]$, there exists a unique integer $h\in [q]$ satisfying $i=h+i'$. In addition, if $i=i'$ we have $h=q$, otherwise $h<q$ always holds. Then we have that each integer $i\in[q-1]$ occurs in $\mathcal{C}^{(k)}+i'$ where $i'\in[q-1]\setminus\{i\}$ exactly $((q-1)^{k-1}+1)/q$ times and $\mathcal{C}^{(k)}+i$ exactly $((q-1)^{k-1}-q+1)/q$ times; integer $q$ occurs in $\mathcal{C}^{(k)}+i'$ where $i'\in[q-1]$ exactly $((q-1)^{k-1}+1)/q$ times.
Thus, each integer $i\in[q-1]$ occurs exactly
\begin{align*} 
(q-2)\frac{(q-1)^{k-1}+1}{q}+\frac{(q-1)^{k-1}-q+1}{q}=\frac{(q-1)^{k}-1}{q}
\end{align*}times in $\mathcal{C}^{(k+1)}$ and $q$ occurs exactly
\begin{align*} 
(q-1)\frac{(q-1)^{k-1}+1}{q}=\frac{(q-1)^{k}+q-1}{q}
\end{align*}times in $\mathcal{C}^{(k+1)}$. That is $|\mathcal{C}_i|=((q-1)^{k}-1)/q$ where $i\in[q-1]$ and $|\mathcal{C}_{q}|=((q-1)^{k}+q-1)/q$ for the case $m=k+1$.  So our statement holds for the case $m=k+1$ when $m$ is odd.

Similarly, if $k$ is odd,  by our hypothesis there are exactly $q-1$ different integers $h_i\in [q]$ satisfying that $|\mathcal{C}_{h_i}|=((q-1)^{k-1}-1)/q$ for each $i\in[q-1]$ and exactly one integer $h_q$ satisfying that $|\mathcal{C}_{h_q}|=((q-1)^{k-1}+q-1)/q$. Without loss of generality, we assume that $|\mathcal{C}_i|=((q-1)^{k-1}-1)/q$ where $i\in[q-1]$ and $|\mathcal{C}_{q}|=((q-1)^{k-1}+q-1)/q$. Similarly, we can also have that each integer $i\in[q-1]$ occurs in $\mathcal{C}^{(k)}+i'$ where $i'\in[q-1]\setminus\{i\}$ exactly $((q-1)^{k-1}-1)/q$ times and $\mathcal{C}^{(k)}+i$ exactly $((q-1)^{k-1}+q-1)/q$ times; integer $q$ occurs in $\mathcal{C}^{(k)}+i'$ where $i'\in[q-1]$ exactly $((q-1)^{k-1}-1)/q$ times. So each integer $i\in[q-1]$ occurs exactly
\begin{align*}
(q-2)\frac{(q-1)^{k-1}-1}{q}+\frac{(q-1)^{k-1}+q-1}{q}=\frac{(q-1)^{k}+1}{q}
\end{align*}times in $\mathcal{C}^{(k+1)}$ and $q$ occurs exactly
\begin{align*} 
(q-1)\frac{(q-1)^{k-1}-1}{q}=\frac{(q-1)^{k}-q+1}{q}
\end{align*}times in $\mathcal{C}^{(k+1)}$. That is $|\mathcal{C}_i|=((q-1)^{k}+1)/q$ where $i\in[q-1]$ and $|\mathcal{C}_{q}|=((q-1)^{k}-q+1)/q$ for the case $m=k+1$.  So our statement holds for the case $m=k+1$ when $m$ is even. Then the proof is completed.

\bibliographystyle{IEEEtran}
\bibliography{reference}
\end{document}